\documentclass[a4paper,11pt]{article}
\usepackage{amssymb,amsthm,amsmath,mathrsfs}
\usepackage{cite}
\usepackage[colorlinks,linkcolor=blue,citecolor=magenta]{hyperref}
\usepackage{graphicx}
\usepackage{algorithm,algpseudocode}
\usepackage{caption}
\usepackage [a4paper, footskip=1cm, headheight = 16pt, top=2.5cm, bottom=2.8cm,  right=2.5cm,  left=2.5cm ]{geometry}
\usepackage{lipsum,environ}
\usepackage{tabulary}
\usepackage{tabularx}
\usepackage{multirow}
\usepackage{booktabs}

\newtheorem{pro}{Proposition}

\newtheorem{lemma}[pro]{Lemma}
\newtheorem{theorem}[pro]{Theorem}
\newtheorem{claim}[pro]{Claim}

\newtheorem{prerem}[pro]{Remark}
\newtheorem{predefin}[pro]{Definition}

\newenvironment{probc}[3]{\vskip 5pt
	\noindent {\bf #1}\\[3pt]
	\hspace*{-6pt}\begin{tabular}{p{60pt}l}
		{INSTANCE:}& \parbox[t]{12cm}{#2}\\[3pt]
		QUERY:    & \parbox[t]{12cm}{#3}}{\end{tabular}\\[4pt]}
\renewcommand{\qedsymbol}{$\blacksquare$}
\DeclareMathOperator{\MAM}{\alpha^\prime_{ac}}
\DeclareMathOperator{\dom}{\mathscr{D}}
\DeclareMathOperator{\range}{\mathscr{R}}
\DeclareMathOperator{\tw}{tw}
\begin{document}

\title{\vspace*{2cm}	On the Parameterized Complexity\\ of the Acyclic Matching Problem 
}
\author{
		Sahab Hajebi%
	\thanks{Department of Mathematics,
		Isfahan University of Technology,
		P.O. Box: 84156-83111, Isfahan, Iran. Email Address: \href{mailto:sahab.hajebi@gmail.com}{sahab.hajebi@gmail.com}.} 
	\and
		Ramin Javadi%
		\thanks{Corresponding author, Department of Mathematical Sciences,
			Isfahan University of Technology,
			P.O. Box: 84156-83111, Isfahan, Iran. School of Mathematics, Institute for Research in Fundamental Sciences (IPM), P.O. Box: 19395-5746,
			Tehran, Iran.  Email Address: \href{mailto:rjavadi@iut.ac.ir}{rjavadi@iut.ac.ir}.}
		\thanks{This research was in part supported by a grant from IPM (No. 1400050420).}
}
\date{}
\maketitle
\begin{abstract}
A matching is a set of edges in a graph with no common endpoint. A matching $M$ is called acyclic if the induced subgraph on the endpoints of the edges in $M$ is acyclic. Given a graph $G$ and an integer $k$, Acyclic Matching Problem seeks for an acyclic matching of size $k$ in $G$. The problem is known to be NP-complete. In this paper, we investigate the complexity of the problem in different aspects.  First, we prove that the problem remains NP-complete for the class of planar bipartite graphs of maximum degree three and arbitrarily large girth. Also, the problem remains NP-complete for the class of planar line graphs with maximum degree four. Moreover, we study the parameterized complexity of the problem. In particular, we prove that the problem is W[1]-hard on bipartite graphs with respect to the parameter $k$. 
On the other hand, the problem is fixed parameter tractable with respect to the parameters $tw$ and $(k,c_4)$, where $tw$ and $c_4$ are the treewidth and the number of cycles with length $4$ of the input graph. We also prove that the problem is fixed parameter tractable with respect to the parameter $k$ for the line graphs and every proper minor-closed class of graphs (including planar graphs).
\\
\begin{itemize}
\item[]{{\footnotesize {\bf Keywords:}\ Acyclic Matching, Induced Matching, Computational Complexity, Parameterized Complexity.}}
\item[]{ {\footnotesize {\bf Subject classification:} 05C70, 05D15.}}
\end{itemize}
\end{abstract}
	\section{Introduction}
Throughout the paper, all graphs are simple, undirected and loopless. Given a graph $G=(V,E)$, a subset of pairwise nonadjacent edges $M\subseteq E$ is called a \textit{matching} of $G$. 
A matching $M$ of $G$ is called an \textit{acyclic matching} if  $G[M]$ is acyclic, where $G[M]$ is the subgraph of $G$ induced on the all endpoints of $M$. Also, a matching $M$ of $G$ is called an \textit{induced matching} if $G[M]$ contains no edges except the edges in $M$.
It is clear that every induced matching is an acyclic matching. 
In the \textsc{Acyclic (Induced) Matching} problem, the input is a graph $G$ and a positive integer $k$ and the task is to find an acyclic (induced) matching of size $k$ in $G$.

The \textsc{Induced Matching} problem was first introduced by Stockmeyer and Vazirani \cite{vazir} as a variant of the maximum matching problem and proved to be NP-complete in general graphs. This problem is known to be NP-complete for planar graphs of maximum degree $4$ \cite{pd4IM}, bipartite graphs of maximum degree three, $C_4$-free bipartite graphs \cite{BiIM}, $r-$regular graphs for $r \geq 5$, line-graphs, chair-free graphs, and Hamiltonian graphs \cite{LinIM}.  The problem is known to be polynomial time solvable for many classes of graphs such as trees \cite{Zito}, chordal graphs  \cite{cameron} and line-graphs of Hamiltonian graphs \cite{LinIM} (for more information see e.g. \cite{weekchor,Arc,trapez,ATfree,p5free}). On the parameterized complexity of the problem, it was shown that the problem is W[1]-hard with respect to the parameter $k$ on bipartite graphs \cite{PIM}, on $K_{1,4}$-free graphs \cite{clawfree}, and also on Hamiltonian bipartite graphs \cite{HamW1IM}. For the latter result, it was shown that the problem cannot be solved in time $n^{o(\sqrt{k})}$, where $n$ is the number of vertices of the input graph, unless the 3-SAT problem can be solved in subexponential time. On the other hand, it becomes fixed-parameter tractable with respect to the parameter $k$ on planar graphs, graphs with girth at least 6, line graphs, graphs with bounded treewidth \cite{PIM} and also on claw-free graphs \cite{clawfree}.

The \textsc{Acyclic Matching} problem first introduced by Goddard et al. \cite{1} and is proved to be NP-complete in general graphs.  
In this paper, we focus on the complexity of the  \textsc{Acyclic Matching} problem and for the first time investigate the parameterized complexity of the problem with respect to some known parameters such as $k$, $\Delta$ and treewidth and on some subclasses of graphs such as bipartite graphs, line graphs and proper minor-closed graphs. First, we survey most known results about the problem.

In \cite{panda}, Panda and Pradhan proved that the \textsc{Acyclic Matching} problem is NP-complete for planar bipartite graphs and Perfect elimination bipartite graphs.  Here, we will improve this result and prove that the problem remains NP-complete for planar perfect elimination bipartite graphs with maximum degree three and girth at least $l$, for every fixed integer $l\geq 3$ (see Theorem~\ref{theorem:1}).
In \cite{chaud}, Panda and Chaudhary proved that the \textsc{Acyclic Matching} problem is NP-complete for
comb-convex bipartite graphs and dually chordal graphs. They also proved that the problem is hard to approximate within a ratio of $n^{1-\epsilon}$ for any $\epsilon>0$, unless P = NP and this problem is APX-complete for $2k + 1$-regular graphs
for $k \geq3$.

The \textsc{Acyclic Matching} problem is known to be polynomial time solvable for the class of bipartite permutation graphs, chain graphs \cite{panda},  $P_4-$free graphs, $2P_3-$free graphs \cite{Onsome}, chordal graphs \cite{chordal}, split graphs and proper interval graphs \cite{chaud}. In \cite{Onsome}, Furst and Rautenbach proved that for a given graph $G$, deciding that the size of the maximum acyclic matching of $G$ is equal to the size of the maximum matching of $G$ is NP-hard even for bipartite graphs with a perfect matching and the maximum degree $4$. This result contrasts with a result by Kobler and Rotics \cite{LinIM} that deciding if the size of the maximum induced matching of $G$ is equal to the size of the maximum matching is polynomially solvable for general graphs. Also,  Furst and Rautenbach et al. \cite{aprox} proved that every graph with $n$ vertices, the maximum degree $\Delta$, and no isolated vertex, has an acyclic matching of size at least $(1-o_{_\Delta}(1))\, {6n}/{\Delta^2}$, and they explained how to find such an acyclic matching in polynomial time. Moreover, they provided a ${(2\Delta+2)}/{3}$-approximation algorithm for the  \textsc{Acyclic Matching} problem, based on greedy and local search strategies.\\

\begin{table}[t]
	\centering\footnotesize
	\caption{\label{tbl:results}Overview of new NP-hardness and parameterized results for the \textsc{Acyclic Matching} problem. Herein, $c_4$ and $tw$ denote the  number of cycles with length four and the treewidth of the input graph, respectively. }
	\vskip 7pt
	\begin{tabular}{ p{.15\textwidth} p{.65\textwidth}}     
		\toprule
		Parameter & Result \\%

		\midrule
		-&
		NP-hard for planar perfect elimination bipartite graphs with maximum degree three and girth at least $l\geq3$ (Theorem~\ref{theorem:1})\\
		\cmidrule{1-2}
		-&
		NP-hard for planar line graphs with maximum degree $4$ (Theorem~\ref{theorem:2})\\
		\cmidrule{1-2}
		$k$ &
		W[1]-hard on bipartite graphs (Theorem~\ref{theorem:4})\\
		\cmidrule{1-2}
		$k$ &
		FPT for line graphs (Theorem~\ref{theorem:7})\\
		\cmidrule{1-2}
		$k$ &
		FPT for every proper minor-closed class of graphs (Theorem~\ref{theorem:14})\\
		\cmidrule{1-2}
		$tw$ &
		FPT  (Theorem~\ref{theorem:8})\\
		\cmidrule{1-2}
		$(k,c_4)$ &
		FPT  (Theorem~\ref{theorem:12})\\
		\bottomrule
	\end{tabular}
\end{table}

In this paper, we prove some new hardness results about the \textsc{Acyclic Matching} problem and investigate its parameterized complexity. Our main results are summarized in Table~\ref{tbl:results}. 
The organization of forthcoming sections is as follows. In Section~\ref{sec:not}, we give the necessary notations and definitions which are needed later. In Section~\ref{sec:hard}, we provide our hardness results for the \textsc{Acyclic Matching} Problem including the fact that the problem remains NP-complete for the class of planar perfect elimination bipartite graphs with maximum degree three and girth of arbitrary large as well as the class of planar line graphs with maximum degree four. We also prove that the problem is W[1]-hard on bipartite graphs with respect to the parameter $k$.
In Section~\ref{sec:fpt}, we prove that the problem is fixed parameter tractable with respect to the parameters $tw$ and $(k,c_4)$, where  $tw$ and $c_4$ denote the treewidth and the number of cycles with length four of the input graph, respectively.
Moreover, we prove that the problem is fixed parameter tractable with respect to $k$ in the class of line graphs and also every proper minor-closed class of graphs (including planar graphs).

\section{Notations and Conventions}\label{sec:not}
For a function $f$, the domain and the range of $f$ are denoted by $\dom(f)$ and $\range(f)$, respectively. For a subset $X\subseteq \dom(f)$, the restriction of $f$ on $X$ is denoted by $f|_X$. 
Given a positive integer $k$, the notation $[k]$ stands for the set $\{1,\ldots,k \}$.
The set of neighbors of a vertex $v\in V(G)$ in the graph $G$ is denoted by $N_G(v)$ and the closed neighborhood of $v$ in $G$ is defined as $N_G[v]=N_G(v)\cup \{v\}$. The degree of a vertex $v$ in $G$ is defined as $\deg_G(v)=|N_G(v)|$.  For a subset of vertices $S\subseteq V(G)$, $N_G(S)$ and $N_G[S]$ are respectively defined as $\cup_{v\in S} N_G(v)\setminus S$ and $\cup_{v\in S} N_G[v]$. Also, we drop the subscript $G$ whenever there is no ambiguity. 
Two nonadjacent vertices $u$ and $v$ in $V(G)$ is called \textit{twin} if $N_G(u)=N_G(v)$.
An edge is called a \textit{pendant edge} if one of its endpoints has degree equal to one. 
For two disjoint subsets of vertices $A,B\subset V(G)$, we say that $A$ \textit{is complete} (resp. \textit{incomplete}) to $B$, if every vertex in $A$ is adjacent (resp. nonadjacent) to every vertex in $B$. Let $A$ and $B$ be two disjoint subsets of vertices in the graph $G$ such that $|A|=|B|$ and let $f:A\to B$ be a bijection. We say that \textit{$A$ is anti-matched to $B$} by  $f$ in $G$ if every vertex $v\in A$ is adjacent to every vertex in $B$ except $f(v)$. An \textit{anti-matching} is a bipartite graph $G$ with the bipartition $(X,Y)$ such that $X$ is anti-matched to $Y$ in $G$ by a bijection $f$.

A matching $M$ in a graph $G$ is a set of edges in $E(G)$ which are mutually nonadjacent (i.e. have no common endpoint). The set of all endpoints of the edges in the matching $M$ is denoted by $V(M)$ and the subgraph of $G$ induced on $V(M)$ is denoted by $G[M]$. We say that $M$ saturates all vertices in $V(M)$. Also, for a subset of vertices $S\subset V(G)$, the induced subgraph of $G$ on $S$ is denoted by $G[S]$. An \textit{induced} (resp. \textit{acyclic}) \textit{matching} is a matching $M$ such that $G[M]$ has exactly $|M|$ edges (resp. is acyclic).
The size of the maximum matching, the maximum induced matching and the maximum acyclic matching are denoted by $\alpha^\prime(G)$, $\alpha^\prime_{\text{in}}(G)$ and $\MAM(G)$, respectively.
The \textsc{Acyclic (Inudced) Matching} problem is defined as follows.

\begin{probc}
	{Acyclic Matching (AM)}{A graph $G$ and integer $k$.}{Is there an acyclic matching $M$ in $G$ such that $|M|=k$? }
\end{probc}	
\begin{probc}
	{Induced Matching (IM)}{A graph $G$ and integer $k$.}{Is there an induced matching $M$ in $G$ such that $|M|=k$?}
\end{probc}

The line graph of $G$ is a graph denoted by $L(G)$ whose vertices are corresponding to the edges of $G$ and the vertices $e$ and $e'$ in $L(G)$ are adjacent if their corresponding edges in $G$ are adjacent. A graph $G$ is called a line graph if there is a graph $H$ such that $G=L(H)$. 

Let $n, m$ be two integers with $n, m \geq 2$.
An \textit{$(n \times m)$-grid} is the graph $G_{n\times m}$ with $V(G_{n\times m}) = [n] \times [m]$ and
$E(G_{n\times m}) = 
\{\{(i_1, j_1),(i_2, j_2)\}\ : \ |i_1-i_2|+|j_1-j_2| = 1,\ i_1, i_2 \in [n],\ j_1, j_2 \in [m]\}$.
A \textit{grid} graph is defined as an induced subgraph of an $(n \times m)$-grid for some integers $m,n\geq2$.  
An \textit{elementary $(n \times m)$-wall} is a graph $G =
(V, E)$ with the vertex set
\begin{align*}
V =&\{(1, 2j- 1)\ :\ 1 \leq j \leq m\}
\ \cup
\ \{
(i, j) \ :\ 1 < i < n, 1 \leq j \leq 2m\}
\ \cup \\
&\{
(n, 2j -1) \ :\ 1 \leq j \leq m, \ \text{if $n$ is even}\}\ \cup
\ \{
(n, 2j) \ : \ 1 \leq j \leq m, \ \text{if $n$ is odd}\}   
\end{align*}
and the edge set
\begin{align*}
E =&\{\{(1, 2j - 1),(1, 2j + 1)\} \ : \ 1 \leq j < m\} \ \cup \ \{\{(i, j),(i, j + 1)\}\ :\ 2 \leq i < n, 1 \leq j < 2m\}\\
&\cup \ \{\{(n, 2j),(n, 2j + 2)\}\ :\ 1 \leq j < m \ \text{if $n$ is odd}\}\\	
&\cup \ \{\{(n, 2j - 1),(n, 2j + 1)\}\ :\ 1 \leq j < m \ \text{if $n$ is even}\}\\
&\cup \ \{\{(i, j),(i + 1, j)\}\ :\ 1 \leq i < n, 1 \leq j \leq 2m, \ \text{either $i, j$ are both odd, or $i, j$ are both even}\}
\end{align*}
A \textit{subdivision} of a graph $G$ is a graph which is obtained from $G$ by replacing each edge of $G$ by a path of arbitrary length. An \textit{$(n \times m)$-wall} is a subdivision of an elementary $(n \times m)$-wall. A graph $G$ is called \textit{chordless} if every cycle in $G$ is an induced subgraph of $G$. \\

For the definition of \textit{tree decomposition} and \textit{treewidth} see the standard textbooks e.g. \cite{cygan}. In this paper, we use a modified version of tree decomposition called  \textit{nice tree decomposition}, which is
 defined as a pair  $\mathcal{T} = (T, \{ X_{i} \}_{i\in V(T )} )$ such that $T$ is a binary tree rooted at a vertex $r$ with the following properties. 
 
 \begin{enumerate}
 	
 	\item $X_{r} = \emptyset$ and $X_{l} = \emptyset$ for every leaf $l$ of $T$. In other words, all the leaves as well as the root possess empty bags.
 	\item Every non-leaf node of $T$ is of one of the following four types:
 	\begin{itemize}
 		\item \textbf{Introduce vertex node:} a node $i$ with exactly one child $j$ such that $X_{i} = X_{j} \cup \{v\}$ for some vertex $v \notin X_{j}$; we say that $v$ is introduced at $i$.
 		\item \textbf{Introduce edge node:} a node $i$, labeled with an edge $uv\in E(G)$ such  that $u, v \in X_i$, and with exactly one child $j$ such that $X_i = X_{j}$; we say
 			that edge $uv$ is introduced at $i$.
 		\item \textbf{Forget node:} a node $i$ with exactly one child $j$ such that $X_{i} = X_{j} \setminus \{w\}$ for some vertex $w \in X_{j}$; we say that $w$ is forgotten at $i$.
 		\item \textbf{Join node:} a node $i$ with two children $j,k$ such that $X_{i} = X_{j} = X_{k}$.
 	\end{itemize}
 	
 \end{enumerate}

An algorithm that transforms in linear time a tree decomposition into a nice one of the same treewidth is presented in \cite{kloks1994vol}.
The width of a tree decomposition is the size of its largest bag $X_i$ minus one. The treewidth $\tw(G)$ of a graph $G$ is the minimum width among all possible tree decompositions of $G$.

\section{Hardness Results}\label{sec:hard} Let $G=(X, Y, E)$  be a bipartite graph. An edge $e=xy$ is said to be a \textit{bisimplicial} edge if $G[N[\{x,y\}]]$ is a complete bipartite subgraph of $G$. Let $\mathcal{M}=(x_1y_1,x_2y_2,\ldots,x_ky_k)$ be an ordering of  pairwise nonadjacent edges  of $G$. Let $A_j=\{x_1,x_2,\ldots,x_j\}\cup \{y_1,y_2,\ldots,y_j\}$ and let $A_0=\emptyset$. The ordering $\mathcal{M}$ is called a \textit{perfect edge elimination} ordering of G if $x_{j+1}y_{j+1}$ is a bisimplicial edge in $G[(X\cup Y)\backslash A_j]$ for $j=0,1,\ldots,k-1$ and $G[(X\cup Y)\backslash A_k]$ has no edge. A bipartite graph for which there exists a perfect edge elimination ordering is called a \textit{perfect elimination bipartite graph}. Panda and Pradhan \cite{panda} showed that the \textsc{Acyclic Matching} problem is NP-hard for perfect elimination bipartite graphs. In this section, we strengthen their result and prove that the problem remains NP-hard for planar perfect elimination bipartite graphs with maximum degree three and girth at least $l$ for any arbitrary positive integer $l$.
Our reduction is from the following problem which is known to be NP-hard for $4$-regular planar graphs \cite{11,stck} (in fact, it is the dual of the \textsc{Feedback Vertex Set} Problem). 

\begin{probc}
	{Induced Acyclic Subgraph (IAS)}{A graph $G$ and integer $k$.}{Is there a subset of vertices $S\subseteq V(G)$ such that $|S|=k$ and $G[S]$ is acyclic?}
\end{probc}	

\begin{theorem}
	For every integer $l\geq 3$, the \textsc{Acyclic Matching} problem is NP-hard for planar perfect elimination bipartite graphs with maximum degree three and girth at least $l$.\label{theorem:1}
\end{theorem}
\begin{proof}
	We prove the theorem by a polynomial reduction from \textsc{Induced Acyclic Subgraph} problem that is known to be NP-hard for 4-regular planar graphs \cite{11}. Let $(G(V,E),k)$ be an instance of \textsc{Induced Acylcic Subgraph} problem where $G$ is a 4-regular planar graph on $n$ vertices and $m$ edges. Now, fix an ordering on the edge set $E$. Also, fix an integer $l\geq 3$.
	We are going to construct a graph $G'$ as follows. 
	
	\begin{figure}
		\includegraphics[width=\textwidth,
		height=\textheight,
		keepaspectratio,
		trim=0 35 0 0pt,]{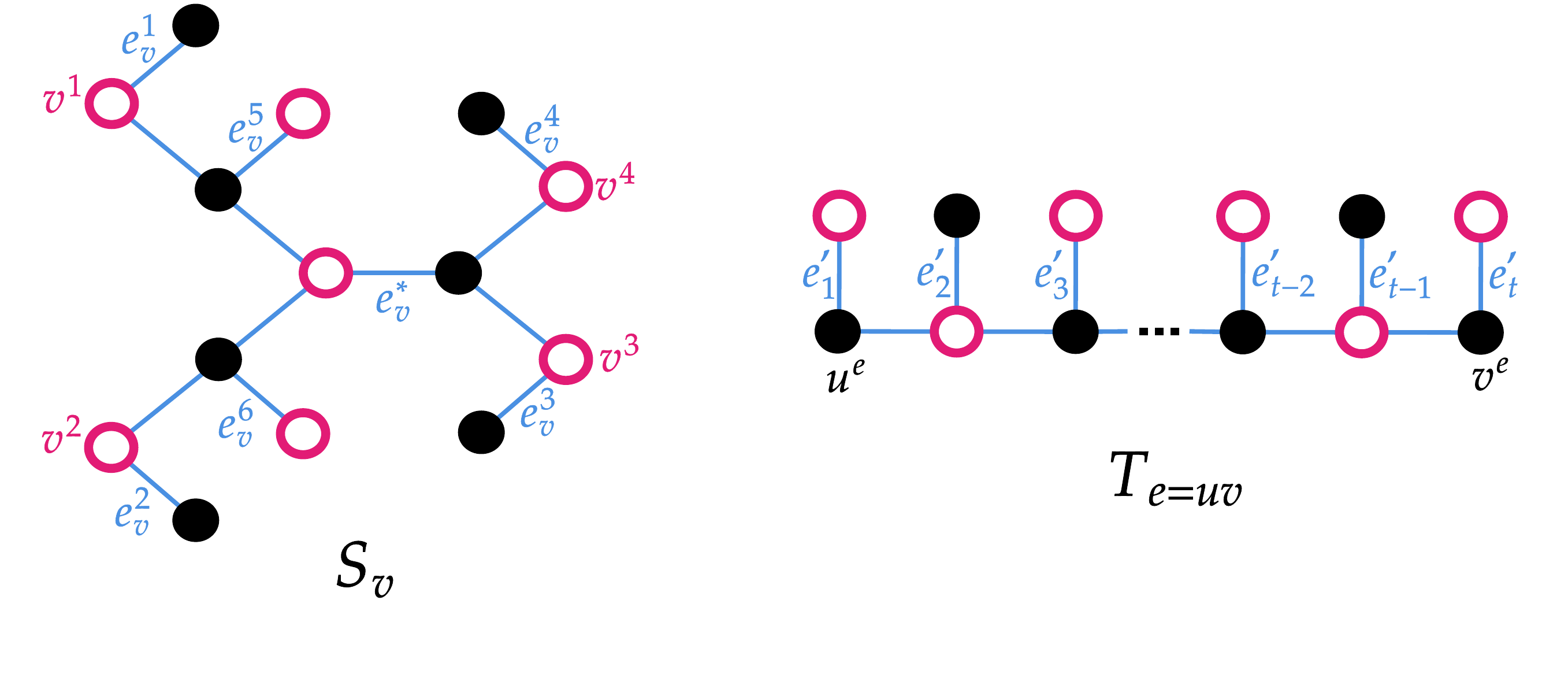}
		\caption{The vertex gadgets $S_v$ and the edge gadgets $T_e$}\label{fig:1a}		
	\end{figure}
	
	For every vertex $v \in V$, consider a vertex gadget $ S_v$  and for every edge $e \in E$ consider an  edge gadget $T_e$ as in Figure~\ref{fig:1a}, where $t=2l-3$. Then, the vertices of $G'$ is the disjoint union of vertices of all gadgets $S_v$, $v\in V$ and $T_e$, $e\in E$.  If $e=uv$ is an edge of $G$ such that $e$ is the $i$th edge in the set of edges incident with $u$ and $j$th edge in the set of edges incident with $v$, then connect $u^i$ in $S_u$ to $u^e$ in $T_e$ and connect  $v^j$ in $S_v$ to $v^e$ in $T_e$. The edge set of $G'$ is the set of all these edges along with the edges inside all gadgets $S_v$, $v\in V$ and $T_e$, $e\in E$. 		 
	
	Now we claim that $G$ has an induced acyclic subgraph on $k$ vertices if and only if $G'$ has an acyclic matching of size $tm+6n+k$. In order to prove the claim, let $H$ be an induced acyclic subgraph of $G$ on $k$ vertices. Define,
	
	\[
	M:=\{e'_1,\ldots, e'_t\,:\,e\in E(G)\}\cup\{e^1_v\,,\,e^2_v\,,\,e^3_v\,,\,e^4_v\,,\,e^5_v\,,\,e^6_v\,:\,v\in V(G)\}\cup\{e^*_v\,:\,v\in V(H)\}.
	\] 
	
	It is clear that $M$ is a matching of size $tm+6n+k$. Suppose that there exists a cycle $C$ in $G'[M]$. So $C$ is crossing through some gadgets $S_{v_1}\,,\,T_{v_1v_2}\,,\,S_{v_2}\,\dots\,S_{v_{h}}\,,\,T_{v_{h}v_1}\,,\,S_{v_1}$ for some vertices $v_1,\dots ,v_h\in V(G)$. Now, we have $v_i\in V(H)$, for all $i\in [h]$, since otherwise $e^*_{v_i}\notin M$ and it is not possible that $C$ passes through $S_{v_i}$. Therefore, $v_1, v_2, \ldots, v_h$ form a cycle in $H$ which is a contradiction. Thus, $M$ is an  acyclic matching of size $tm+6n+k$. Conversely, let $M$ be an acyclic matching of size $tm+6n+k$ in $G'$.
	\begin{claim}\label{clm1}
		There is an acyclic matching $M'$ of size at least $tm+6n+k$ in $G'$ such that for every $e\in E(G)$, $e'_1,\ldots,e'_t\in M'$.
	\end{claim}
	\begin{proof}[Proof of Claim~\ref{clm1}]\renewcommand{\qedsymbol}{} We construct $M'$ from $M$.
		Let $E^*=\{e^*_v:\ v\in V(G) \}$. First, note that one can assume that every edge $e\in M\setminus E^*$ is a pendant edge, since one can replace every edge in $M\setminus E^*$ with its adjacent pendant edge and the matching remains acyclic. 
		
		Now,  consider a gadget $T_e$, where $e=uv\in E(G)$ and  $e_h'\notin M$, for some $h\in[t]$. Then, add $e_h'$ to $M$ and if both $u^i\in V(S_u)$ and $v^j\in V(S_v)$ are covered by the matching $M$, then remove the edge in $M$ which covers $u^i$. The obtained matching $M'$ remains acyclic. This proves the claim.
	\end{proof}
	
	Now, consider the matching $M'$. It has $t$ edges from each edge gadget, so it has at least  $6n+k$ edges from vertex gadgets. The size of the maximum matching of a vertex gadget is $7$, therefore there are at least $k$ vertex gadgets $S_{v_1},\dots , S_{v_k}$ that $M'$ contains all edges $e^1_{v_i}\,,\,e^2_{v_i}\,,\,e^3_{v_i}\,,\,e^4_{v_i}\,,\,e^5_{v_i}\,,\,e^6_{v_i}$ and $e^*_{v_i}$, $1\leq i\leq k$. Let $H$ be the induced subgraph  of $G$ on the vertex set $\{v_1,\dots,v_k\}$. We claim that $H$ is acyclic. If $C=v_{i_1},\dots,v_{i_h}$ is a cycle in $H$, then $C$ can be extended to a cycle in $G'[M']$ that crosses through the gadgets $S_{v_{i_1}}$, $T_{v_{i_1}v_{i_2}}$, $S_{v_{i_2}}$, $\ldots$, $S_{v_{i_h}}$, $T_{v_{i_h}v_{i_1}}$, $S_{v_{i_1}}$, because $G'[M']$ contains all these gadgets. This contradiction implies that $H$ is acyclic. This completes the reduction. 
	
	Also, note that the maximum degree of $G'$ is three. Since $G$ is planar, obviously $G'$ is also planar. Moreover, all cycle lengths in $G'$ are even and at least $3t+9$ which is greater than $l$. So, $G'$ is bipartite with girth at least $l$. Finally, $G'$ is a perfect elimination bipartite graph. To see this, note that $\tilde{M}=\{e'_h, e\in E(G), h\in [t]\}\cup \{ e_v^i, v\in V(G), i\in [6]\}\cup \{ e_v^*, v\in V(G)\}$ is a perfect matching of $G'$. Consider an ordering of $\tilde{M}$ as $\mathcal{M}=(x_1y_1,\ldots, x_ry_r)$ such that all edges in $\{e^*_v, v\in V(G)\}$ lie after all remaining edges in $\tilde{M}$. For each $i\in [r]$, define $A_i=\{x_1,x_2,\ldots,x_i\}\cup \{y_1,y_2,\ldots,y_i\} $. One can see that for each $i\in [r]$, the edge $x_{i+1}y_{i+1}$ is a bisimplicial edge in $G'_i=G'[V(G')\backslash A_i]$ because $G'_i[N[\{x_{i+1},y_{i+1}\}]$ is a star. Also, $G'[V(G')\backslash A_r]$ has no edge because $\tilde{M}$ is a perfect matching. Thus, $\mathcal{M}$ is a perfect edge elimination ordering of $G'$.
\end{proof}
It is proved in \cite{2} that the \textsc{Induced Acyclic Subgraph} problem is W[1]-hard with respect to the parameter $k$. By a simple reduction from this problem (i.e. adding pendant edges to each vertex of the input graph) one can prove that the \textsc{Acyclic Matching} problem is also W[1]-hard with respect to the parameter $k$. Here, we improve this result and prove that the problem remains W[1]-hard on bipartite graphs.

\begin{theorem}
	The \textsc{Acyclic Matching} problem for bipartite graphs $ ($AMB$) $ is W[1]-hard with respect to the parameter $k$.\label{theorem:4}
\end{theorem}
In order to prove the above theorem, first we need a couple of easy observations.
\begin{pro}\label{lem6}
	Let $G$ be an anti-matching on at least four vertices. Then $\MAM(G)= 2$.
\end{pro}
\begin{proof}
	It is easy to see that $\MAM(G)\geq 2$. Now, let $M$ be an acyclic matching of $G$. If $|M|\geq 3$, then $G[M]$ is a supergraph of an anti-matching on $6$ vertices which  contains a cycle $C_6$. This contradiction implies that $|M|\leq 2$.
\end{proof}

A \textit{blow-up} of a graph $G$ is a graph $G'$ which is obtained from $G$ by a finite sequence of the following operation: take a vertex $v\in V(G)$ and replace $v$ with a nonempty stable set of new vertices $A_v$ and connect every vertex in $A_v$ to every neighbor of $v$ in $G$.
\begin{pro}\label{lem7}
	Let $G$ be a graph and $G'$ be a blow-up of $G$. Then, $\MAM(G)= \MAM(G')$.
\end{pro}
\begin{proof}
	Let $G'$ be obtained from $G$ by blow-up of a vertex $v\in V(G)$. 
	It is clear that $G$ is an induced subgraph of $G'$ and thus $\MAM(G)\leq \MAM(G') $.
	Now, let $M$ be the maximum acyclic matching of $G'$ and let $H=G'[M]$. If $V(H)$ contains more than one vertex in $A_v$, then $H$ contains a cycle $C_4$ (since all vertices in $A_v$ are twin). Thus, $|V(H)\cap A_v|\leq 1$ and so $H$ is an induced subgraph of $G$. This shows that $\MAM(G)\geq |M|=\MAM(G')$.
\end{proof}

\begin{proof}[Proof of Theorem~\ref{theorem:4}]
	We give a three-stage parameterized reduction. First, consider the following four problems. 
	
	\begin{probc}{Acyclic Matching on Bipartite Graphs (AMB)}
		{A bipartite graph $G$ and a positive integer $k$.}{Is there an acyclic matching of size $k$ in $G$?}
	\end{probc}
	
	\begin{probc}{Multicolor Acyclic Matching on Bipartite Graphs (MAMB)}
		{A bipartite graph $G$ and a $k$-partition $\mathcal{P}=(V_1,\ldots, V_k)$ of $V(G)$.}{Is there a multicolor acyclic matching $M$ for $(G,\mathcal{P})$? (A matching $M=\{u_1v_1,\ldots, u_kv_k\}$ in $G$ is called  multicolor if for every $i\in [k]$, $u_i,v_i\in V_i$.)}
	\end{probc}

	\begin{probc}
		{Irredundant Set (IS)}{A graph $G$ and integer $k$.}{Is there an irredundant set in $G$ of size $k$? (An irredundant set in $G$ is a set $S\subset V(G)$ such that for every $u\in S$, $N[S\setminus \{u\} ] \neq N[S]$.)}
	\end{probc}

	\begin{probc}
		{Multicolor Irredundant Set (MIS)}{A graph $G$ and a $k$-partition $\mathcal{P}=(V_1,\ldots, V_k)$ of $V(G)$.}{Is there a multicolor irredundant set $S$ for $(G,\mathcal{P})$? (A multicolor irredundant set for $(G,\mathcal{P})$ is a set $S\subset V(G)$ such that for every $i\in [k]$,  $|S\cap V_i|=1$ and for $u\in S\cap V_i$, there is a vertex $v\in V_i$ such that $v\in N[S]\setminus N[S\setminus \{u\}]$.)}
	\end{probc}
	
	We provide a sequence of parameterized reductions as follows $\text{IS}\leq \text{MIS} \leq \text{MAMB} \leq \text{AMB}$. In all problems the parameter is $k$. Downey et al. \cite{downey} proved that the \textsc{Irredundant Set} problem is W[1]-hard.

	\paragraph{Parameterize reduction from IS to MIS.}
	Let $(G,k)$ be an instance of IS. Construct a graph $G'$ with the vertex set $V(G')= \{u_i: u\in V(G), i\in [k] \}$ such $u_i$ is adjacent to $u_j$ for every $u\in V(G)$ and $i,j\in [k], i\neq j,$ and if $u$ is adjacent to $v$ in $G$, then $u_i$ is adjacent to $v_j$ for every $i,j\in [k]$. Also, for every $i\in [k]$, let $V_i=\{u_i: u\in V(G)\}$ and $k'=k$. Then $(G',k')$ with the partition $(V_1,\ldots,V_k)$ of $V(G')$ is an instance of MIS.  If $(G,k)$ is a yes-instance with an irredundant set $S=\{u^1,\ldots,u^k \}$, then we claim that $S'=\{u^1_1,\ldots,u^k_k\}$ is a multicolor irredundant set of $G'$. To see this, note that for every $i\in[k]$, there is a vertex $v^i\in N_G[S]\setminus N_G[S\setminus \{u^i\}]$. Thus, $v^i_i\in N_{G'}[S']\setminus N_{G'}[S'\setminus\{u_i^i\}]$. On the other hand, suppose that $(G',k')$ is a yes-instance of MIS. Then, there is a multicolor irredundant set  $S'=\{u_1^1,\ldots,u_k^k \}$ in $G'$. First, note that $S=\{u^1,\ldots, u^k\}$ is a set of size $k$ in $V(G)$, otherwise if $u^i=u^j$ for $i\neq j$, then $N_{G'}[u^i_i]=N_{G'}[u^j_j]$ which is a contradiction with the fact that $S'$ is an irredundant set. Now, for every $i\in [k]$, there is a vertex $v^i_i\in N_{G'}[S']\setminus N_{G'}[S'\setminus\{u_i^i\}]$. Therefore, $v^i\in N_G[S]\setminus N_G[S\setminus \{u^i\} ]$. Hence, $S$ is an irredundant set of size $k$ for $G$ and $(G,k)$ is a yes-instance for IS.

	\paragraph{Parameterized reduction from MIS to MAMB.}
	Let $(G,\mathcal{P})$ be an instance of MIS with the partition $\mathcal{P}=(V_1,\ldots,V_k)$ of $V(G)$. Construct the bipartite graph $G'$ where $V(G')$ is the disjoint union of the set $\{u^i: u\in V(G), i\in\{1,2\}  \}$ and $\{z_i,w_i: i\in [k]\}$. For every vertices $u,v\in V(G)$, if $u\in N[v]$, then $u^i$ is adjacent to $v^j$ for every $i,j\in\{1,2\}$, $i \neq j$. Also, for every $i\in[k]$, $z_i$ is adjacent to all vertices in $\{u^1: u\in V_i\}$. Finally, $w_k$ is adjacent to $z_k$ and for every $i\in [k-1]$, $w_i$ is adjacent to $z_i$ and $z_{i+1}$.
	Now, consider the $2k$-partition $(V'_1,\ldots,V'_{2k})$ of $V(G')$ such that for $i\in [k]$, $V'_i=\{u^j: u\in V_i, j\in\{1,2\}\}$ and for $i\in [2k]\setminus [k]$, $V'_i=\{z_{i-k},w_{i-k}\}$. First, suppose that $S=\{u_1,\ldots u_k\}$ is a multicolor irredundant set for $(G,\mathcal{P})$, where for each $i\in[k]$, $u_i\in V_i$ and there is a vertex $v_i\in V_i \cap (N[S]\setminus N[S\setminus\{u_i\}] )$. Now, consider the multicolor matching $M= \{u_i^1v_i^2: i\in [k]\}\cup \{z_iw_i: i\in[k]\}$ of size $2k$ in $G'$. Note that $G'$ induces a tree on the set  $\{u_i^1,z_i,w_i: i\in [k]\}$. So, if there is a cycle in $G'[M]$, then there must be an edge between $u_i^1$ and $v_j^2$ for some $i\neq j$ which is impossible since $v_j\not\in N_G[S\setminus\{u_j\}]$. Hence, $M$ is a multicolor acyclic matching. On the other hand, suppose that $M$ is arbitrary multicolor acyclic matching in $G'$. Since for each $i\in [2k]\setminus [k]$, $V'_i$ is a set of size two, $\{z_1w_1,\ldots,z_kw_k\}\subset M$. Also, for every $i\in[k]$,  $M$ contains an edge $e_i$ with both endpoints in $V'_i$. So, $e_i=u_i^1v_i^2$, for some $u_i,v_i\in V_i$. We claim that $S=\{u_1,\ldots, u_k\}$ is a multicolor irredundant set in $G$. To see this, note that for each $i\in[k]$, $v_i\in N[u_i]$. If $v_i\in N[u_j]$, for some $j\neq i$, then $v_i^2$ and $u_j^1$ are adjacent in $G'$ and so $G'[M]$ contains a cycle which is a contradiction. Hence, $v_i\not\in N[S\setminus \{u_i\}]$ and $S$ is a multicolor irredundant set for $(G,\mathcal{P})$.

	\paragraph{Parameterized reduction from MAMB to AMB.}
	This is the main part of the proof. Let the bipartite graph $G=(V,E)$ with the bipartition $(X,Y)$ and the $k$-partition $\mathcal{P}=(V_1,\ldots,V_k)$ of $V$ be an instance of MAMB. 
	For each $i\in[k]$, let $E_i$ be the set of all edges of $G$ with both endpoints in $V_i$. 
	We are going to construct a bipartite graph $G'=(V',E')$ and choose an integer $k'=11k^2-9k+(ck^2+1)(5k^2-3k)$, where $c$ is a constant that will be determined in the proof,  such that $G$ has a multicolor acyclic matching of size $k$ if and only if $G'$ contains an acyclic matching of size $k'$. First, note that for each $i\in[k]$, $V_i$ intersects both parts $X$ and $Y$, otherwise $(G,\mathcal{P})$ is a No instance. In the construction of $G'$, we use the following three gadgets (To illustrate the construction, see Figure~\ref{fig3} as an example).
	
	\begin{figure}
		\includegraphics[width=\textwidth,
		height=\textheight,
		keepaspectratio,
		trim=0 30 0 0pt]{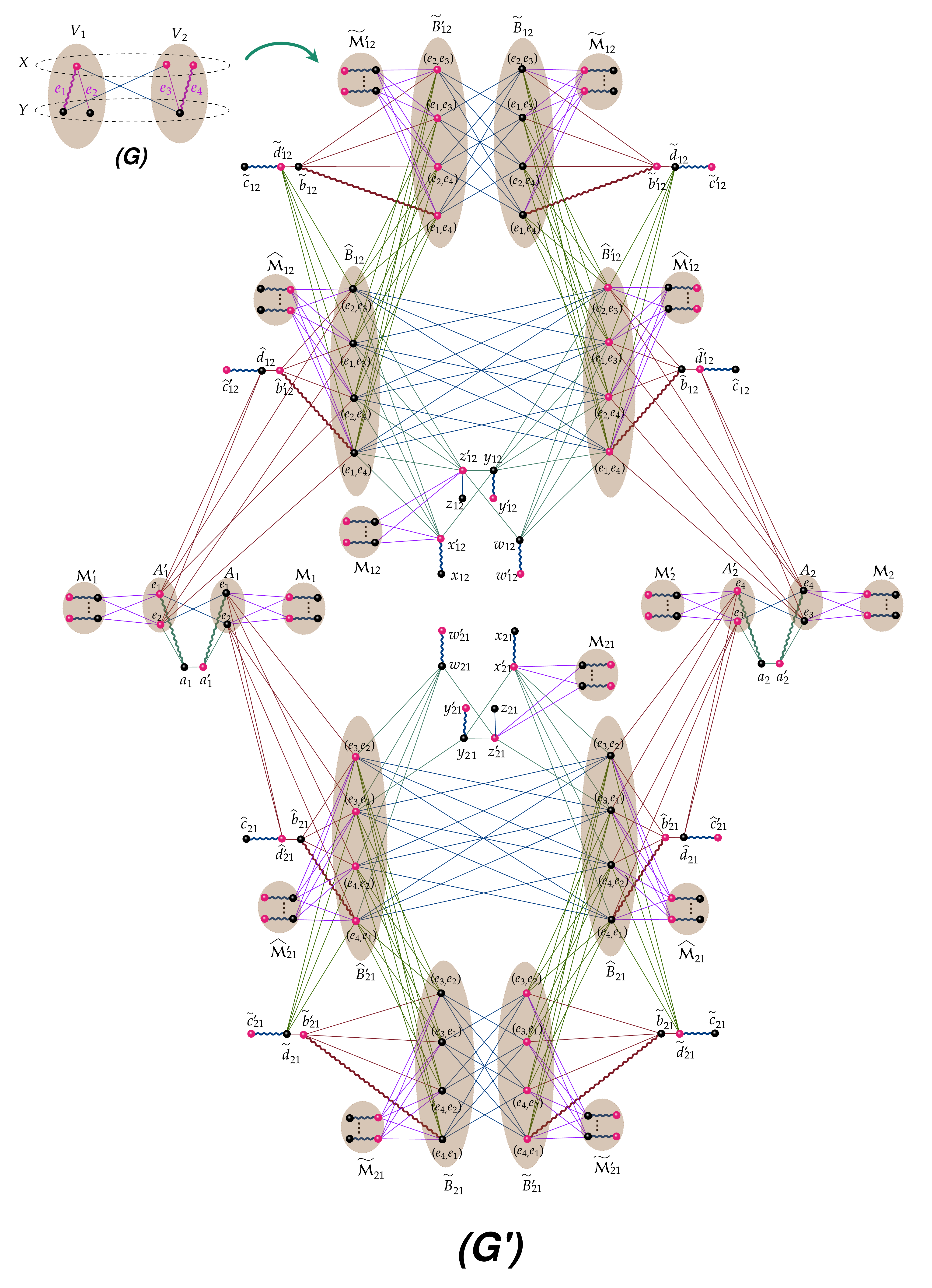}
		\caption{An example of the construction of $G'$ from $G$.}\label{fig3}
	\end{figure}
	\begin{itemize}
		\item $\mathcal{M}=\mathcal{M}(N,L)$ is an induced matching where $(N,L)$ is the bipartition of $V(\mathcal{M})$.
		\item $\mathcal{G}^1$ is the bipartite graph with the bipartition $(A\cup\{a\},A'\cup\{a'\})$ where $|A|=|A'|$. The vertices $a$ and $a'$ are respectively complete to $A'\cup \{a'\}$ and $A\cup\{a\} $ and  $A$ is anti-matched to $A'$.
		\item $\mathcal{G}^2$ is the bipartite graph with the bipartition $(B\cup\{b,c,d\},B'\cup\{b',c',d'\})$, where $|B|=|B'|$ and $f:B\to B'$ is a bijection. The vertices $b$ and $b'$ are respectively complete to $B'\cup \{d'\}$ and $B\cup\{d\} $, $c$ and $c'$ are respectively adjacent to $d'$ and $d$, and $B$ is anti-matched to $B'$ by $f$. 
		\item $\mathcal{G}^3$ is the bipartite graph comprised with the union of two disjoint copies of $\mathcal{G}^2$, say $\hat{\mathcal{G}}^2=\hat{\mathcal{G}}^2(\hat{B}\cup\{\hat{b},\hat{c},\hat{d}\},\hat{B}'\cup\{\hat{b}',\hat{c}',\hat{d}'\})$ with bijection $\hat{f}:\hat{B}\to \hat{B}'$  and $\tilde{\mathcal{G}}^2=\tilde{\mathcal{G}}^2(\tilde{B}\cup\{\tilde{b},\tilde{c},\tilde{d}\},\tilde{B}'\cup\{\tilde{b}',\tilde{c}',\tilde{d}'\})$  with bijection $\tilde{f}:\tilde{B}\to \tilde{B}'$ and the following additional edges. Let $f:\hat{B}\to \tilde{B}'$ be a bijection. Then, $\hat{B}$ is anti-matched to $\tilde{B}'$ by $f$ and $\tilde{B}$ is anti-matched to $\hat{B}'$ by the bijection $\hat{f}of^{-1}o\tilde{f}$. Also, $\tilde{d}$ and $\tilde{d}'$ are respectively complete to $\hat{B}'$ and $\hat{B}$.
		\item $\mathcal{G}^4$ is the bipartite graph on the vertex set $\{x,y,z,w,x',y',z',w'\}$ with edges $xx'$, $yy'$, $zz'$, $ww'$, $x'y$, $yz'$ and $z'w$.
		
	\end{itemize}
	
	It should be noted that in the following construction, for instance if $\mathcal{G}^1_i$ is a copy of $\mathcal{G}^1$, by abuse of notation, we use $A_i,A'_i, a_i,a'_i$ for the (sets of) vertices in $\mathcal{G}^1_i$ corresponding to $A,A',a,a'$.  
	The graph $G'$ is comprised with some vertex-disjoint copies of the above gadgets, where there are a number of edges between these gadgets, as defined in the following.

	First, for every part $V_i$ in $G$, set a disjoint copy $\mathcal{G}^1_i$ of $\mathcal{G}^1$, where $|A_i|=|A_i'|=|E_i|$ and  every vertex in $A_i$ and its counterpart in $A'_i$ are corresponding to an edge in $E_i$. 
	For each edge $e\in E_i$, let us call its corresponding vertices in $A_i$ and $A'_i$ by $[e,A_i]$ and $[e,A'_i]$ respectively. 
	
	Also, for every pair $(V_i,V_j)$, $i,j\in [k]$, $i\neq j$, let $\mathcal{G}^3_{ij}$ and $\mathcal{G}^4_{ij}$ be the disjoint copies of $\mathcal{G}^3$ and $\mathcal{G}^4$ respectively, where $|\hat{B}_{ij}|=|\hat{B}'_{ij}|=|\tilde{B}_{ij}|=|\tilde{B}'_{ij}|=|E_i\times E_j|$ and every vertex  $v\in \hat{B}_{ij}$ and its three counterparts $\hat{f}_{ij}(v)\in \hat{B}'_{ij}$, $f_{ij}(v)\in \tilde{B}'_{ij}$ and $\tilde{f}_{ij}^{-1}o f_{ij}(v)\in \tilde{B}_{ij}$ are corresponding to a pair $(e_i,e_j)$ for some $e_i\in E_i$ and $e_j\in E_j$.  To simplify notations, for each pair $(e_i,e_j)\in E_i\times E_j $, let us call its corresponding vertex in $\hat{B}_{ij}$ by $[(e_i,e_j),\hat{B}_{ij}]$. Similar notations are used for its corresponding vertices in $\hat{B}'_{ij},\tilde{B}_{ij}$ and $\tilde{B}'_{ij}$.
	
	Let $F$ be the disjoint union of all $\mathcal{G}^1_i$, $i\in[k]$, all $\mathcal{G}^3_{ij}$ and $\mathcal{G}^4_{ij}$,  $i,j\in [k]$, $i\neq j$ with the following additional edges: 
	For each $i,j\in [k]$, $i\neq j$, the connections between $\mathcal{G}^1_i$ and $\mathcal{G}^3_{ij}$ are as follows; $\hat{d}_{ij}$ and $\hat{d}'_{ij}$ are respectively complete to $A_i'$ and $A_j$. Also, every vertex $[e_j,A_j]$ in $A_j$ corresponding to the edge $e_j\in E_j$, is adjacent to all vertices $[(e,e'),\hat{B}'_{ij}]$ in $\hat{B}'_{ij}$ for every $e\in E_i$ and $e'\in E_j\setminus\{e_j\}$. Moreover, every vertex $[e_i,A'_i]$ in $A'_i$ corresponding to the edge $e_i\in E_i$, is adjacent to all vertices $[(e,e'),\hat{B}_{ij}]$ in $\hat{B}_{ij}$ for every $e\in E_i\setminus\{e_i\}$ and $e'\in E_j$. Now, for each $i,j\in [k]$, $i\neq j$, the connections between  $\mathcal{G}^3_{ij}$ and  $\mathcal{G}^4_{ij}$ are as follows; for every pair $(e_i,e_j)\in E_i\times E_j$, if the endpoint of $e_i$ in $X$ is adjacent to the endpoint of $e_j$ in $Y$, then the vertex $[(e_i,e_j),\hat{B}_{ij}]$ (resp. $[(e_i,e_j),\hat{B}'_{ij}]$) is adjacent to the vertices $x'_{ij}$ and $z'_{ij}$ (resp. $w_{ij}$ and $y_{ij}$) in $\mathcal{G}^4_{ij}$, otherwise, the vertex $[(e_i,e_j),\hat{B}_{ij}]$ (resp. $[(e_i,e_j),\hat{B}'_{ij}]$) is adjacent to the vertex $x'_{ij}$ (resp. $w_{ij}$) in $\mathcal{G}^4_{ij}$. There is no more edges in $F$.
	
	Now, let $G'$ be obtained from $F$ as follows. Let $\mathcal{M}$ be the gadget defined as above such that $|N|=|L|=ck^2+1$ (we will determine the constant $c$ shortly in Claim~\ref{cc}). For each $i\in [k]$, add two disjoint copies of $\mathcal{M}$, say $\mathcal{M}_i$ and  $\mathcal{M}'_i$, to $F$ and connect every vertex in $L_i$ (resp. $L'_i$) to $A_i$ (resp. $A'_i$). Also, for each $i,j\in [k]$, $i\neq j$, add five disjoint copies of $\mathcal{M}$, say $\hat{\mathcal{M}}_{ij}$, $\hat{\mathcal{M}}'_{ij}$, $\tilde{\mathcal{M}}_{ij}$, $\tilde{\mathcal{M}}'_{ij}$ and $\mathcal{M}_{ij}$,  to $F$ and connect every vertex in $\hat{L}_{ij}$, $\hat{L}'_{ij}$, $\tilde{L}_{ij}$, $\tilde{L}'_{ij}$ and $L_{ij}$ respectively to every vertex in $\hat{B}_{ij}$, $\hat{B}'_{ij}$, $\tilde{B}_{ij}$, $\tilde{B}'_{ij}$ and $\{x'_{ij},z'_{ij}\}$.

	Now, suppose that $G$ contains a multicolor acyclic matching $M=\{e_1,\ldots,e_k\}$, where $e_i\in E_i$, for each $i\in[k]$. Now, we construct an acyclic matching of size $k'$ in $G'$. Define the following subsets of $E(G')$.
	
	\begin{equation}\label{Mi}
	\begin{split}
	M_1 &= \{\{a'_i, [e_i,A_i]\}, \{a_i, [e_i,A'_i]\} \ :\  i\in[k] \},\\
	M_2 &= \{\{\hat{d}_{ij}, \hat{c}'_{ij}\}, \{\hat{d}'_{ij}, \hat{c}_{ij}\},
	\{\hat{b}_{ij}, [(e_i,e_j),\hat{B}'_{ij})]\}, \{\hat{b}'_{ij}, [(e_i,e_j),\hat{B}_{ij})]\} \ :\  i,j\in[k], i\neq j \},\\
	M_3 &= \{
	\{\tilde{d}_{ij}, \tilde{c}'_{ij}\}, \{\tilde{d}'_{ij}, \tilde{c}_{ij}\}, \{\tilde{b}_{ij}, [(e_i,e_j),\tilde{B}'_{ij})]\}, \{\tilde{b}'_{ij}, [(e_i,e_j),\tilde{B}_{ij})]\} \ : \ i,j\in[k], i\neq j \},\\
	M_4 &= \{\{x'_{ij},x_{ij}\}, \{y'_{ij},y_{ij}\}, \{w'_{ij},w_{ij}\}  \ : \ i,j\in[k], i\neq j  \}.
	\end{split}
	\end{equation}
	
	Now, let $M'$ be the set of edges obtained from  $M_1\cup M_2 \cup M_3\cup M_4 $ by adding all edges in $\mathcal{M}_{i}$, $\mathcal{M}'_{i}$, $\mathcal{M}_{ij}$, $\hat{\mathcal{M}}_{ij}$, $\hat{\mathcal{M}}'_{ij}$, $\tilde{\mathcal{M}}_{ij}$ and $\tilde{\mathcal{M}}'_{ij}$, $i,j\in[k]$, $i\neq j$. Thus,
	\[
	|M'|=2k+11k(k-1)+(ck^2+1)(2k+5k(k-1))=k'.
	\] 
	\begin{figure}[t]
		\includegraphics[width=\textwidth,
		height=\textheight,
		keepaspectratio,
		trim=0 0 0 0pt]{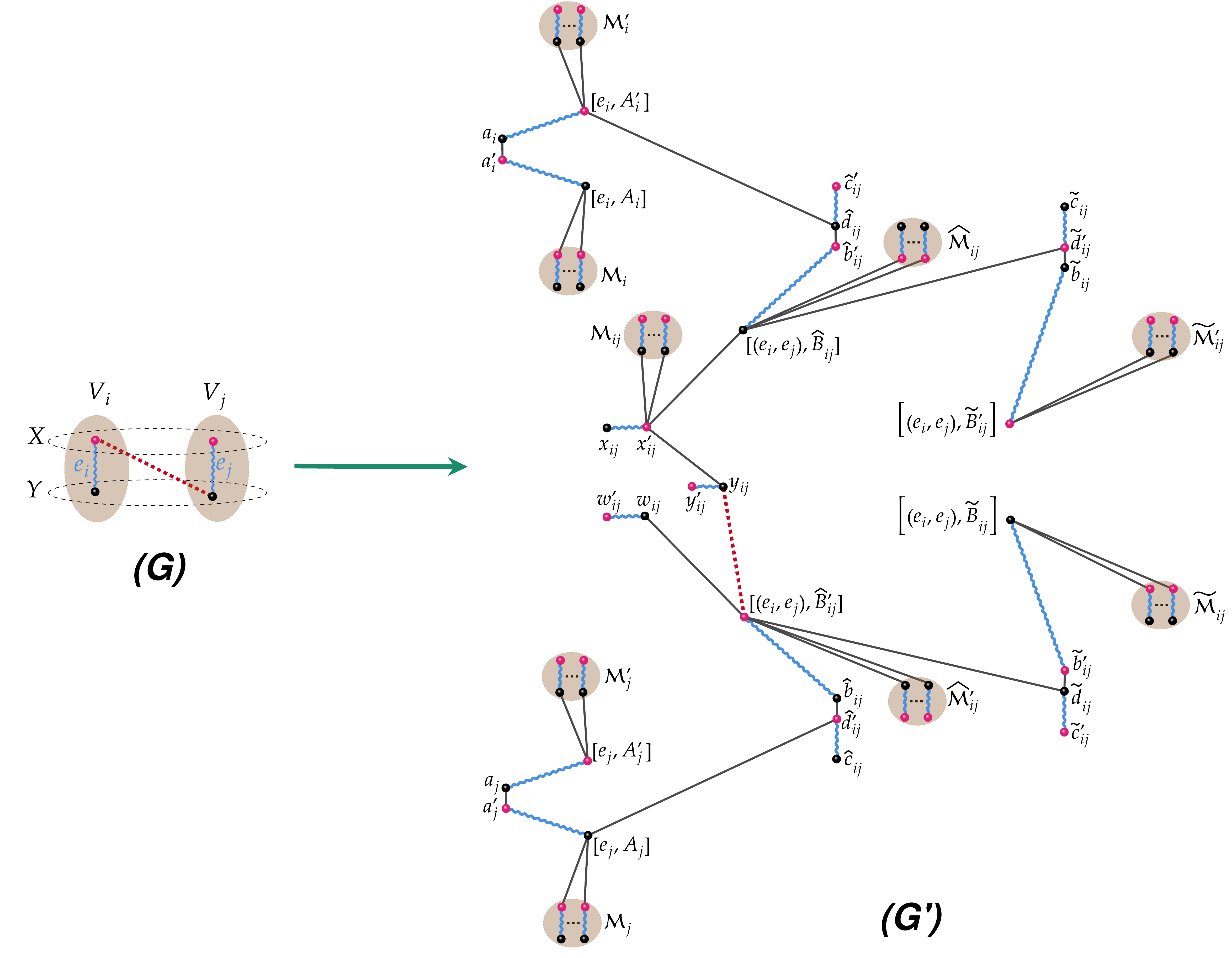}
		\caption{How paths in $G$ can be translated to paths in $G'$. } \label{fig}
	\end{figure}
	For every $i,j\in [k]$, $i\neq j$, as you might see in Figure~\ref{fig}, if the endpoint of $e_i$ in $X$ is connected to the endpoint of $e_j$ in $Y$, then there is a path from $[e_i,A'_i]$ to $[e_j,A_j]$ in $G'$. Thus, every cycle in $G'[M']$ is evidently corresponding to a cycle  in $G[M]$ and vice versa. Thus, $M'$ is an acyclic matching in $G'$.
	
	Now, suppose that $M'$ is the maximum acyclic matching of size at least $k'$. First, we prove some claims. Let
	
	\[\overline{\mathcal{M}}:=\cup_{i\in [k]} (\mathcal{M}_i\cup \mathcal{M}'_i) \cup_{i\neq j\in [k]} (\hat{\mathcal{M}}_{ij}\cup \hat{\mathcal{M}}'_{ij}\cup\tilde{\mathcal{M}}_{ij}\cup \tilde{\mathcal{M}}'_{ij}\cup \mathcal{M}_{ij}). \]
	
	\begin{claim} \label{cc}
		Let $G''$ be the graph obtained from $G'$ by deleting all vertices in $\overline{\mathcal{M}}$. Then there is a constant $c$ such that $\MAM(G'')\leq ck^2$.
	\end{claim}
	
	To see the claim, let $G'''$ be the induced subgraph of $G'$ on the vertices $\cup_i (A_i\cup A'_i)\cup_{i\neq j} (\hat{B}_{ij}\cup \hat{B}'_{ij}\cup \tilde{B}_{ij}\cup \tilde{B}'_{ij})$. Then $E(G''')$ is partitioned into the following set of edges,
	\begin{align*}
	&E(A_i,A'_i), i\in [k],\\
	&E(A'_i,\hat{B}_{ij}), E(A_i,\hat{B}'_{ji}), E(\hat{B}_{ij},\hat{B}'_{ij}), E(\hat{B}'_{ij},\tilde{B}_{ij}), E(\tilde{B}_{ij},\tilde{B}'_{ij}), E(\tilde{B}'_{ij},\hat{B}_{ij}), i,j\in[k], i\neq j   
	\end{align*}
	and the subgraph induced on each of these sets is either an anti-matching or a blow-up of an anti-matching whose maximum acyclic matching is at most two by Propositions~\ref{lem6} and \ref{lem7}. Thus, $\MAM(G''')=O(k^2)$. Also, $G''$ is obtained from $G'''$ by adding $O(k^2)$ vertices. Hence, there exists some constant $c$ such that $\MAM(G'')\leq ck^2$.
	This proves Claim~1.
	
	\begin{claim}\label{M}
		The matching $M'$ contains all edges in  $\overline{\mathcal{M}}$. 
	\end{claim}
	
	To see the claim, first note that $M'$ should contains at least two edges from each of the graphs $\mathcal{M}_i, \mathcal{M}'_i$, $i\in [k]$, $\hat{\mathcal{M}}_{ij}$, $\hat{\mathcal{M}}'_{ij}$, $\tilde{\mathcal{M}}_{ij}$, $\tilde{\mathcal{M}}'_{ij}$, $\mathcal{M}_{ij}$,  $i, j\in [k]$, $i\neq j$. For the contrary, suppose that $M'$ contains at most one edge from say $\mathcal{M}_i$. Then, $|M'|\leq \MAM(G'')+(5k^2-3k)(ck^2+1)-ck^2 \leq (5k^2-3k)(ck^2+1)$, where the last inequality is due to Claim~\ref{cc}. Thus, $|M|<k'$ which is a contradiction. 
	So, for each of the matching gadgets in $\overline{\mathcal{M}}$ say $\mathcal{M}_i$, $M'$ contains at least two edges in $\mathcal{M}_i$. Now, every two vertices in $L_i$ have the same neighbors in $G''$, thus, since $M'$ is an acyclic matching, every two vertices in $L_i$ have at most one neighbor in $V(G'')\cap V(M')$. Also, every edge in $M'$ between $V(\overline{\mathcal{M}})$ and $V(G'')$ can be replaced with an edge in $\overline{\mathcal{M}}$. Therefore, we may assume that $M'$ contains all edges in $\overline{\mathcal{M}}$. This proves the claim.

	\begin{claim}\label{leaf}
		The intersection of $V(M')$ with each of the sets $A_i$, $A'_i$, $\hat{B}_{ij}$, $\hat{B}'_{ij}$, $\tilde{B}_{ij}$, $\tilde{B}'_{ij}$ and $\{x'_{ij}, z'_{ij}\}$, $i,j\in[k], i\neq j$, has size at most one. 
	\end{claim}
	
	To see the claim, first note that by Claim~\ref{M}, $M'$ contains all edges in $\overline{\mathcal{M}}$. So, if $V(M')$ contains two vertices in one of the above sets, say $A_i$, then since $A_i$ is complete $L_i$, $G'$ induces a cycle $C_4$ on $V(M')$ which is a contradiction. This proves the claim.\\
	
	We partition the edge set of $G''$ into the following four sets.
	
	\begin{align*}
	E'_1&:=\bigcup_{i\in [k]} E(A_i\cup\{a_i\},A'_i\cup\{a_i'\}),\\
	E'_2&:= \bigcup_{i\neq j\in [k]} E( \hat{B}_{ij}\cup \{\hat{d}_{ij}\},A_i'\cup \hat{B}'_{ij} \cup \{\hat{c}'_{ij},\hat{b}'_{ij}\} ) \cup E(A_j\cup \{\hat{c}_{ij},\hat{b}_{ij}\}, \hat{B}'_{ij}\cup \{\hat{d}'_{ij}\})\\
	E_3'&:=\bigcup_{i\neq j\in [k]} E(\hat{B}_{ij},\tilde{B}'_{ij}\cup\{\tilde{d}'_{ij}\})\cup E(\tilde{B}_{ij}\cup\{\tilde{d}_{ij}\},\hat{B}'_{ij})\cup E(\tilde{B}_{ij}\cup\{\tilde{b}_{ij},\tilde{c}_{ij},\tilde{d}_{ij}\},\tilde{B}'_{ij}\cup\{\tilde{b}'_{ij},\tilde{c}'_{ij},\tilde{d}'_{ij}\})\\
	E_4'&:= \bigcup_{i\neq j\in [k]} E(\hat{B}_{ij}\cup\{x_{ij},y_{ij},z_{ij},w_{ij}\},\{x'_{ij},y'_{ij},z'_{ij},w'_{ij}\}) \cup E(\{y_{ij},w_{ij}\},\hat{B}'_{ij})
	\end{align*}
	
	Due to Claims~\ref{M} and \ref{leaf}, one can easily see that 
	\begin{align*}
	|M'\cap  E'_1|&\leq 2k,\\
	|M'\cap  E'_2|&\leq 4k(k-1),\\
	|M'\cap  E'_3|&\leq 4k(k-1),\\
	|M'\cap  E'_4|&\leq 3k(k-1).
	\end{align*}
	
	By Claim~\ref{M}, the number of edges in $M'$ with both endpoints in $V(G'')$ is at least $11k^2-9k$. Thus, equality holds in all above inequalities. On the other hand, the only acyclic matching of size $2k$ in $E'_1$ is of the form $M_1$ defined in \eqref{Mi}, for some edges $e_i\in E_i$, $i\in [k]$. So, by Claim~\ref{leaf}, we have 
	\begin{align*}
	M'\cap E'_1 &= \{\{a'_i, [e_i,A_i]\}, \{a_i, [e_i,A'_i]\} \ :\  i\in[k] \},\\
	M'\cap E'_2&= \{\{\hat{d}_{ij}, \hat{c}'_{ij}\}, \{\hat{d}'_{ij}, \hat{c}_{ij}\},
	\{\hat{b}_{ij}, \hat{u}'_{ij}\}, \{\hat{b}'_{ij}, \hat{u}_{ij}\} \ :\  i,j\in[k], i\neq j \},\\
	M'\cap E'_3&=
	\{
	\{\tilde{d}_{ij}, \tilde{c}'_{ij}\}, \{\tilde{d}'_{ij}, \tilde{c}_{ij}\}, \{\tilde{b}_{ij}, \tilde{v}'_{ij}\}, \{\tilde{b}'_{ij}, \tilde{v}_{ij}\} \ : \ i,j\in[k], i\neq j \},
	\end{align*}
	for some vertices $\hat{u}_{ij}\in \hat{B}_{ij}$, $\hat{u}'_{ij}\in \hat{B}'_{ij}$, $\tilde{v}_{ij}\in \tilde{B}_{ij}$ and $\tilde{v}'_{ij}\in \tilde{B}'_{ij}$,  $i,j\in[k], i\neq j$.
	Now, we are going to prove that $\hat{u}_{ij}=[(e_i,e_j),\hat{B}_{ij}]$, $\hat{u}'_{ij}=[(e_i,e_j),\hat{B}'_{ij}]$, $\tilde{v}_{ij}=[(e_i,e_j),\tilde{B}_{ij}]$ and $\tilde{v}'_{ij}=[(e_i,e_j),\tilde{B}'_{ij}]$.

	Let $\hat{u}_{ij}=[(e_i^*,e_j^*),\hat{B}_{ij}]$ and $\hat{u}'_{ij}=[(e_i^\#,e_j^\#),\hat{B}'_{ij}]$. Note that $e_i^*=e_i$, since otherwise $G'$ induces the cycle $C_4$ on vertices $[e_i,A_i'],\hat{d}_{ij}, \hat{b}_{ij}', \hat{u}_{ij} $. Similarly, we have $e_j^\#=e_j$. Also, note that $\tilde{v}_{ij}= [(e_i^\#,e_j),\tilde{B}_{ij}]$, since  otherwise $G'$ induces the cycle $C_4$ on  vertices $\hat{u}'_{ij}, \tilde{v}_{ij}, \tilde{b}_{ij}', \tilde{d}_{ij}$. Similarly, we have 
	$\tilde{v}'_{ij}= [(e_i,e_j^*),\tilde{B}'_{ij}]$. Now, we can see that $e_i=e_i^\#$ and $e_j=e_j^*$, since otherwise  $G'$ induces the cycle $C_8$ on  vertices $\hat{u}_{ij}, \tilde{d}'_{ij}, \tilde{b}_{ij},\tilde{v}'_{ij},\tilde{v}_{ij}, \tilde{b}_{ij}', \tilde{d}_{ij}, \hat{u}'_{ij}$. Thus, $ M'\cap E'_2=M_2$ and $ M'\cap E'_3=M_3$ as defined in \eqref{Mi}.
	
	Now, we prove that $M'\cap E'_4=M_4$ as defined in \eqref{Mi}. First, note that by Claim~\ref{leaf} and the fact that $|M'\cap  E'_4|= 3k(k-1)$, for each $i,j\in [k], i\neq j$, we have $\{y_{ij},y'_{ij}\}, \{w_{ij},w'_{ij}\} \in M'\cap E'_4$ and we have either $\{x'_{ij},x_{ij}\}$ or $\{z'_{ij},z_{ij}\}$ are in $M'\cap E'_4$. Suppose that $\{z'_{ij},z_{ij}\}\in M'$. We show that one can replace $\{x'_{ij},x_{ij}\}$ with $\{z'_{ij},z_{ij}\}$ in $M'$ and it remains an acyclic matching. To see this, suppose that the new $M'$ contains a cycle $C$ including $x'_{ij}$, then $C$ should contain the subpath $P= \hat{u}_{ij}, x'_{ij},y_{ij}, \hat{u}'_{ij} $. Thus,  $y_{ij}$ is adjacent to $\hat{u}'_{ij}$ and then $z'_{ij}$ is also adjacent to $\hat{u}_{ij}$ and therefore we can replace the subpath $P$ in $C$ with $ \hat{u}_{ij}, z'_{ij},y_{ij}, \hat{u}'_{ij} $. So, $G'$ induces a cycle on $V(M')$ which is a contradiction. This shows that $M'\cap E'_4=M_4$. Hence, $M'=M_1\cup M_2\cup M_3\cup M_4$.
	
	Now, we claim that $M=\{e_1,\ldots, e_k\}$ is a multicolor acyclic matching  in $G$. It is clear that $M$ is a multicolor matching, since $e_i\in E_i$ for each $i\in [k]$. Also, as you see in the argument prior to Claim~\ref{cc}, every cycle in $G[M]$ is corresponding to a cycle  in $G'[M']$ and vice versa. Thus, $M$ is an acyclic matching in $G'$.
\end{proof}

For the last result of this section, we investigate the \textsc{Acyclic Matching} problem on line graphs. First, we need the following lemma.

%
\begin{lemma}\label{lem5}
	Let $H$ be a graph and $G$ be the line graph of $H$. Then, $G$ has an acyclic matching of size $k$ if and only if $H$ contains $t$ vertex-disjoint paths $P_1,\ldots,P_t$, for some positive integer $t$, such that every $P_i$ has an even positive length and $\sum_{i=1}^{t}length(P_i)=2k$.
\end{lemma}
\begin{proof}
	First, suppose that $H$ contains the vertex-disjoint paths $P_1,\ldots, P_t$ of even positive lengths such that $\sum_{i=1}^{t} length(P_i)=2k$. So, the line graph of these paths induces a forest in $G$ consisting of a union of odd paths (paths with odd lengths) whose matching number is equal to $k$. 
	
	Now, Suppose that $M$ is an acyclic matching of $G$. So,  $G[M]$ is a forest. Also, $G[M]$ does not have any vertex with degree greater than two, since otherwise $G$ contains an induced $K_{1,3}$ contradicting with the fact that $G$ is a line graph \cite{beineke}. Thus, $G[M] $ is a disjoint union of paths with odd lengths. 
	The only graph whose line graph is a disjoint union of odd paths is a disjoint union of even paths $P_1,\ldots, P_t$. The matching number of $L(P_i)$ is equal to $length(P_i)/2$, therefore, $\sum_{i=1}^{t} length(P_i)=2k$.
\end{proof}
\begin{theorem}
	The \textsc{Acyclic Matching} problem remains NP-hard when the input is restricted to the planar line graphs with maximum degree $4$.\label{theorem:2}
\end{theorem}
\begin{proof}
	We give a reduction from the \textsc{Hamiltonian Path} problem that is  well-known to be NP-hard when the input graph is a grid graph with maximum degree $3$ \cite{Grid}. Let $H$ be an instance of the \textsc{Hamiltonian Path} problem, where $H$ is a grid graph with $\Delta(H)=3$. Without loss of generality, we assume that the number of vertices of $H$ is an odd number $2k+1$.\footnote{If $H$ has an even number of vertices, then let $v$ be a vertex of degree at most two in $H$ and subdivide an edge incident with $v$. It is clear that $H$ has a Hamiltonian path if and only if the new graph has a Hamiltonian path.}
	Let $G$ be the line graph of $H$. Since $\Delta(H)= 3$, we have $\Delta(G)\leq 4$. It is known that the line graph of a graph $H$ is planar if and only if $H$ is planar, $\Delta(H)\leq 4$, and every vertex of degree $4$ in $H$ is a cut-vertex \cite{Planar}. Since $H$ is planar and $\Delta(H)=3$, $G$ is also planar. 
	
	Now, we claim that $H$ has a Hamiltonian path if and only if $G$ has an acyclic matching of size $k$. If $H$ has a Hamiltonian path, it has a path of length $2k$ as a subgraph. So, by Lemma~\ref{lem5}, $G$ has a acyclic matching of size $k$. Conversely, suppose that $G$ has an  acyclic matching of size $k$. Again by  Lemma~\ref{lem5}, there are vertex-disjoint paths $P_1,\dots,P_t$, for some integer $t$, as a subgraph in $H$ such that $\sum_{i=1}^{t}length(P_i)=2k$. Suppose that the length of $P_i$ is $l_i$. So the number of vertices of the path $P_i$ is $l_i+1$. Since $P_1,\dots,P_t$ are vertex-disjoint we have
	\[
	2k+1\geq \sum_{i=1}^{t}(l_i+1)= \sum_{i=1}^{t}l_i+t= 2k+t.
	\]	
	Thus, $t=1$ and $P_1$ is a path with length $2k$, i.e. $P_1$ is a Hamiltonian path in $H$.
\end{proof}

In the next section, we will prove that the \textsc{Acyclic Matching} problem is fixed-parameter tractable with respect to the parameter $k$ on the class of line graphs.
\section{FPT Results}\label{sec:fpt}
In this section, we prove fixed-parameter tractability of the \textsc{Acyclic Matching} problem on some subclasses of graphs such as line graphs, bounded tree-width graphs, $C_4$-free graphs and any proper minor-closed class of graphs. 

In Theorem~\ref{theorem:2}, we proved that the problem is NP-hard for the planar line graphs with maximum degree $4$. Here, we are going to prove that the problem is FPT with respect to the parameter $k$ for the class of line graphs. To see this, note that by Lemma~\ref{lem5}, every line graph $G=L(H)$ has an acyclic matching of size $k$ if and only if $H$ contains a forest consisting of $t$ vertex-disjoint paths $P_1,\ldots,P_t$, for some positive integer $t$, such that every $P_i$ has an even positive length and $\sum_{i=1}^{t}length(P_i)=2k$. The number of such forests is at most equal to $p(2k)$, the number of partitions of $2k$, which is known that is at most $e^{3\sqrt{2k}}$ \cite{8}. On the other hand, by a result of Alon, Yuster and Zwick \cite{9}, the problem of finding a given forest $F$ with $k$ vertices in a arbitrary graph $G$ on $n$ vertices as a subgraph can be solved in $2^{O(k)}.n^{O(1)}$ time. Therefore, deciding if the line graph $G$ on $n$ vertices has an acyclic matching of size $k$ can be solved in $2^{O(k^{3/2})}.n^{O(1)}$ time. In the following, using the technique of color coding, we improve this observation and prove that the problem can be solved in $2^{O(k)}.n^{O(1)}$ time.

%
%

\begin{theorem}\label{theorem:7} For every line graph $G$ on $n$ vertices, the \textsc{Acyclic Matching} problem can be solved in  $\;2^{O(k)}.n^{O(1)}$ time. 
\end{theorem}
\begin{proof}
	Fix the line graph $G$ on $n$ vertices and integer $k$ as an instance of \textsc{Acyclic Matching} problem. Let $H$ be a graph with vertex set $V$ such that $G=L(H)$ (note that $H$ can be found in linear time \cite{lehot,roso}). Without loss of generality, assume that $H$ has no isolated vertex and $|V|\geq 3k$.
	By Lemma~\ref{lem5}, $G$ has a acyclic matching of size $k$ if and only if  $H$ has a forest $F$ as a subgraph which consists of vertex-disjoint paths $P_1,\ldots,P_t$ where each $P_i$ has an even positive length and $\sum_{i=1}^{t}length(P_i)=2k$. Let us call such a forest a \textit{$k$-path forest}. To check if $H$ contains a $k$-path forest, we use the standard technique of color coding. 
	
	Fix a coloring $c:V\to [3k]$ which colors the vertices of $H$ with $3k$ colors. First, using the following dynamic programming, we prove that deciding if $H$ contains a colorful $k$-path forest\footnote{By a colorful forest, we mean a forest whose all vertices have different colors.} $F$  can be solved in $2^{3k}.k.|V|^{O(1)}$ time.  
	
	For any set $S\subset [3k]$, $|S|\geq 3$, every integer $i\in [k]$ and every vertex $u \in V$, define the Boolean function $f(S,i,u)$, such that $f(S,i,u)=1$ if and only if there exists a colorful $i$-path forest $F$ 
	such that the colors of all vertices in $F$ come from $S$ and $u$ is an endpoint of one of the connected components of $F$. It is clear that the answer is yes if and only if $f([3k],k,u)=1$, for some vertex $u\in V$. The following recursion can compute the function $f$. 
	
	First, note that if $|S|=3$, then $f(S,i,u)=1$ if and only if $i=1$ and there are some vertices $v,w\in V$  such that $S=\{c(w),c(v),c(u)\} $ and $w-v-u$ is a path in $H$. Also, for every $S$ with $|S|\geq 4$, we have $f(S,i,u)=1$ if and only if, either
	\begin{itemize}
		\item there are some vertices $w,v\in V$ such that $c(w),c(v),c(u) $ are distinct members of $S$, $w-v-u$ is a path in $H$ and $f(S\setminus\{c(w),c(v),c(u)\}, i-1, x)=1$ for some vertices $x\in V$, or
		\item there are some vertices $w,v\in V$ such that $c(w),c(v),c(u) $ are distinct members of $S$, $w-v-u$ is a path in $H$ and $f(S\setminus\{c(u),c(v)\}, i-1, w)=1$. 
	\end{itemize}
	To finalize the proof, we follow the standard derandomization using the concept of $(n,k)$-perfect hash family \cite{cygan}. An $(n,k)$-perfect hash family $\mathcal{C}$ is a family of functions from $[n]$ to
	$[k]$ such that for every set $S\subseteq [n]$ of size $k $ there exists a function $c \in \mathcal{C}$ such that for every $1\leq j\leq k$, $|c^{-1}(j)\cap S|=1$. It is proved in \cite{noar} that for any $n, k \geq 1$, one can construct an $(n, k)$-perfect hash family of size $e^kk^{O(\log k)} \log n$ in time $e^kk^{O(\log k)}n \log n$. Now, consider an $(|V|,3k)$-perfect hash family $\mathcal{C}$ and for each coloring $c\in \mathcal{C}$, use the above dynamic programming to check if $H$ contains a colorful $k$-path forest $F$. Hence, the runtime of the whole algorithm is at most  $e^{3k}. {(3k)}^{O(\log 3 k)}.2^{3k}.k.|V|^{O(1)}+e^{3k}{3k}^{O(\log 3k)}|V| \log |V|$. Also, since $H$ has no isolated vertex, we have $|V|\leq 2|E(H)|=2n$. Therefore, the runtime of the algorithm is at most $2^{O(k)}n^{O(1)}$.
	%
\end{proof}
	The well-known theorem by Courcelle \cite{4,5} states that every graph property expressible in the monadic second-order (MSO) logic can be decided in linear time on graphs with bounded treewidth. In fact, the \textsc{Acyclic Matching} problem is expressible in MSO logic as follows. (Note that say by $\mathcal{M}(e)$ we mean $e\in \mathcal{M}$.)
\begin{align*}
\text{MAX}\;( \mathcal{M}):&\ \phi_1(\mathcal{M})\land \neg\phi_2(\mathcal{M})
\\
\phi_1(\mathcal{M}):&\ \bigg[\Big(\forall \mathcal{M}(e); E(e)\Big)\land\Big(\forall \mathcal{M}(e_1)\;\forall \mathcal{M}(e_2)\;\forall V(v); \big(e_1\neq e_2\rightarrow\neg(I(v,e_1)\land I(v,e_2))\big)\Big)\bigg]
\\
\phi_2(\mathcal{M}):&\ \bigg[\exists X; \Big(\forall X(v)\;\big(\exists \mathcal{M}(e); I(v,e)\big)\Big)\land \Big(\forall X(v)\;\exists X(u_1)\;\exists X(u_2)\;\exists E(e_1)\;\exists E(e_2);\\
&\ \Big(v\neq u_1\land v\neq u_2\land u_1\neq u_2\land I(v,e_1)\land I(v,e_2)\land I(u_1,e_1)\land I(u_2,e_2)\Big)\Big)\bigg]
\\
\end{align*}
where $E$ and $V$ are unary relation denoting the vertex and the edge set of the graph respectively, $I$ is the incidence relation and $\mathcal{M}$ is a subset of edges that is a matching (due to $\phi_1(\mathcal{M})$) and  has no cycle on its endpoints (due to $\phi_2(\mathcal{M})$), thereby being an acyclic matching. Therefore, the \textsc{Acyclic Matching} problem is FPT with respect to the treewidth of the input graph. However, in Courcelle's theorem, the function of the treewidth in the running time of the algorithm is huge. Here, we give a dynamic programming which can find the maximum acyclic matching in time $t^{O(t)}\cdot O(n)$, where $t$ and $n$ are respectively the treewidth and the number of vertices of the input graph.
\begin{theorem}\label{thm:treewidth}
	Suppose that the graph $G$ on $n$ vertices and its nice tree decomposition of width $t$ are given. The maximum \textsc{Acyclic Matching} of $G$ can be found in time $t^{O(t)}\cdot O(n)$.\label{theorem:8}
\end{theorem}
\begin{proof}
Suppose that a nice tree decomposition $ (T, \{ X_{i} \}_{i\in V(T )} ) $ for $G$ of width $t$ is given.  For each node $i\in V(T)$, let $T_{i}$ be the subtree of $T$ rooted at $i$ and $G_{i} = (V_{i}, E_{i})$ be the subgraph of $G$ such that $V_i=\bigcup_{j \in V(T_{i})} X_{j}$ and $E_i$ are all edges of $G$ introduced in nodes of $T_i$. 

For each node $i\in V(T)$, by a partial solution of $G_i$, we mean a pair $(M',S)$ such that $M'$ is a matching of $G_i$ and $S$ is a subset of $X_i$ where  $V(M')\cap S=\emptyset$ and $G_i[V(M')\cup S]$ is acyclic. In fact $S$ is the set of vertices which are not still saturated by $M'$, however are reserved to be in the final solution.

Now, let $f:X_i\to\{0,1,2\}$ be a function such that when $f(u)=0$, we mean that $u$ is not selected in a partial solution of $G_i$, i.e. $u\not\in V(M')\cup S$, when $f(u)=1$, we mean that $u$ is saturated by $M'$, i.e. $u\in V(M')$ and finally when $f(u)=2$, we mean that $u$ is not saturated by $M'$ but is reserved, i.e. $u\in S$. 
Also, let $g: f^{-1}(\{1,2\})\to [t+1] $ be a function where $g(u)=g(v)$ means that $u$ and $v$ are in the same connected component of $G_i[V(M')\cup S]$. 
Finally, let $M$ be a matching  of $G_i$ where $V(M)\subseteq f^{-1}(1)$. Now, define the function $\psi$ as follows.

\begin{align}
\psi_i[M,f,g]=& \max |M'|\nonumber \\
&\ \text{s.t.}\ M' \text{ is a matching of }G_i,\nonumber\\
&\qquad M=\{uv\in M': u,v\in X_i \}\text{(i.e. restriction of $M'$ on $X_i$ is $M$)} \nonumber\\
& \qquad \forall x\in X_i;\ f(x)=1 \text{ iff } x\in V(M'),\nonumber\\
&\qquad G_i[V(M')\cup f^{-1}(2)] \text{ is acyclic},\nonumber\\
& \qquad \forall u,v\in \dom(g);\ \nonumber\\
&\qquad g(u)=g(v) \text{ iff there is a path from $u$ to $v$ in } G_i[V(M')\cup f^{-1}(2)].\label{path}
\end{align}
Also, if no such $M'$ exists, then define $\psi_i[M,f,g]=-\infty$.
Since $X_r=\emptyset$, it is clear that the final solution is equal to $\psi_r[\emptyset,\emptyset,\emptyset]$. So, for given $M,f,g$, we find $\psi_i[M,f,g]$ recursively in the following cases.

\paragraph{Introduce vertex node.} Suppose that $i$ is an introduce vertex node with the unique child $j$, where $X_i=X_j\cup \{u\}$. In this case, it is clear that if either $f(u)=1$ or $f(u)=2$ and $|g^{-1}(g(u))|\geq 2$, then $\psi_i[M,f,g]=-\infty$, because  $u$ is an isolated vertex in $G_i$. If either $f(u)=0$ or $f(u)=2$ and $g^{-1}(g(u))=\{u\}$, then $\psi_i[M,f,g]=\psi_j[M,f|_{X_j},g|_{X_j}]$. This can be carried out in $O(t)$ time.

\paragraph{Introduce edge node.}  Suppose that $i$ is an introduce edge node labeled with $uv\in E(G)$ with the unique child $j$, where $X_i=X_j$. So, $E(G_i)=E(G_j)\cup \{uv\}$. If $f(u)f(v)=0$, then $\psi_i[M,f,g]=\psi_j[M,f,g]$, because the edge $uv$ does not contribute in the solution. 
If $f(u)f(v)\neq 0$ and $g(u)\neq g(v)$, then $\psi_i[M,f,g]=-\infty$ because there is an edge between $u$ and $v$ which contradicts \eqref{path} and so there is no feasible solution. 
Now, if $f(u)f(v)\neq 0$, $g(u)=g(v)=\ell$ and $uv\not\in M$, then 
\begin{align}
\psi_i[M,f,g]=&\max \psi_j[M,f,g']\nonumber \\
\begin{split} \label{g'}
&\ \text{s.t. } \forall x\in \dom(g)\setminus g^{-1}(\ell);\ g'(x)=g(x),\\
&\qquad \forall x\in g^{-1}(\ell);\ g'(x)\in \{g'(u),g'(v)\},\\
&\qquad g'(u)\neq g'(v),\\
&\qquad g'(u), g'(v)\not\in \range(g)\setminus \{\ell\}. 
\end{split}
\end{align}
The last equality holds because for every solution $M'$ of $\psi_i[M,f,g]$, by removing the edge $uv$, the number of connected components of $G_i[V(M')\cup f^{-1}(2)]$ increases by one and $u$ and $v$ are in different connected components.  
Finally, if $f(u)f(v)\neq 0$, $g(u)=g(v)=\ell$ and $uv\in M$, then

\begin{align*}
\psi_i[M,f,g]=&1+\max \psi_j[M\setminus \{uv\},f',g']\\
&\ \text{s.t. } g' \text{ satisfies } \eqref{g'}, \\
&\qquad \forall x\in X_i\setminus \{u,v\}, f'(x)=f(x),\\
& \qquad f'(u)=f'(v)=2.
\end{align*}
The last equality holds because for every solution $M'$ of the right side, both vertices $u,v$ are reserved (not saturated) and are in different connected components of $G_j[V(M')\cup f'^{-1}(2)]$, so $M'\cup\{uv\}$ is a solution of the left side. This can be carried out in $2^{O(t)}$ time, because the number of possible functions $g'$ is at most $2^{O(t)}$.

\paragraph{Forget node.} Suppose that $i$ is a forget node with the unique child $j$, where $X_i=X_j\setminus \{u\}$, for some $u\in X_j$. Define $f_{_0},f_{_1}:X_j\to \{0,1,2\}$, $g_\ell: \dom(g)\cup \{u\} \to [t+1]$, $\ell\in [t+1]$, as follows.

\[
	f_{_s}(x)= 
\begin{cases}
f(x) & x\in X_i\\
s & x=u
\end{cases},\ s\in\{0,1\}
,\quad 	g_{_\ell}(x)= 
\begin{cases}
g(x) & x\in \dom(g)\\
\ell & x=u
\end{cases}, \ \ell\in[t+1].
\]
Also, define 
\begin{align*}
\psi^1_i[M,f,g]&=\max_{\ell\in [t+1]} \psi_j[M,f_{_1},g_{_\ell}],\\
\psi^2_i[M,f,g]&=\max_{v\in (N(u)\cap f^{-1}(1))\setminus V(M)} \psi_j[M\cup\{uv\},f_{_1},g_{_{g(v)}}],
\end{align*}
Therefore,
\begin{align*}
\psi_i[M,f,g]=&\max \{\psi_j[M,f_{_0},g], \psi^1_i[M,f,g], \psi^2_i[M,f,g]\}.
\end{align*}

The last equality holds because every solution $M'$ of $\psi_i[M,f,g]$ has three possibilities; if $u\not\in V(M')$, then $M'$ is solution of $\psi_j[M,f_{_0},g]$. If $uv\in M'$, for some $v\notin X_j$, then $M'$ is a solution for $\psi^1_i[M,f,g]$. Finally, If $uv\in M'$ for some $v\in X_i$, then $f(v)=1$ and $M'$ is a solution for $\psi^2_i[M,f,g]$. Here, computation of $\psi_i[M,f,g]$ takes $O(t^2)$ time.

\paragraph{Join node.} Finally, suppose that $i$ is a join node with two childs $j,k$, where $X_i=X_j=X_k$ and $G_i=G_j\cup G_k$. For two functions $f_{_1},f_{_2}:X_i\to \{0,1,2\}$, we say that $(f_{_1},f_{_2})$ is a compatible pair with $f$ if the following conditions hold.

\begin{enumerate}
	\item If $f(x)=0$, then $f_{_1}(x)=f_{_2}(x)=0$, if $f(x)=2$, then $f_{_1}(x)=f_{_2}(x)=2$,
	\item if $f(x)=1$ and $x\in V(M)$, then $f_{_1}(x)=f_{_2}(x)=1$, and
	\item if $f(x)=1$ and $x\not\in V(M)$, then either $f_{_1}(x)=1 $ and $f_{_2}(x)=2$, or $f_{_1}(x)=2 $ and $f_{_2}(x)=1$.
\end{enumerate}

Now, let $g_{_1},g_{_2}: \dom(g)\to [t+1]$ be two functions. We construct a bipartite multigraph $H$ with bipartition $(\range(g_{_1}),\range(g_{_2}))$ in which two vertices $\ell_1\in \range(g_{_1})$ and $\ell_2\in \range(g_{_2})$ are adjacent if $g_{_1}^{-1}(\ell_1) \cap g_{_2}^{-1}(\ell_2)$ is nonempty. Also, we add a parallel edge between $\ell_1$ and $\ell_2$ in $H$ if  $g_{_1}^{-1}(\ell_1) \cap g_{_2}^{-1}(\ell_2)$ is not a subset of a connected component of $G_i[\dom(g)]$.

Now, we say that $(g_{_1},g_{_2})$ is an admissible pair with respect to $g$ if the following conditions hold.

\begin{enumerate}
	\item For any $x,y\in \dom(g)$, if $g_{_1}(x)=g_{_1}(y)$ or $g_{_2}(x)=g_{_2}(y)$, then $g(x)=g(y)$,
	\item $H$ is a forest, and
	\item For any $x,y\in \dom(g)$, $g(x)=g(y)$ iff there is a path between $g_{_1}(x)$ and $g_{_1}(y)$ in $H$. 
\end{enumerate}

Therefore, we have 
\begin{align*}
\psi_i[M,f,g]=&\max \psi_j[M,f_{_1},g_{_1}]+ \psi_k[M,f_{_2},g_{_2}]-|M|\\
&\ \text{s.t. } (f_{_1},f_{_2}) \text{ is a compatible pair with } f\text{, and} \\
&\qquad (g_{_1},g_{_2}) \text{ is an admissible pair w.r.t. } g.
\end{align*}

To see correctness of last equality, let $M'$ be a solution of $\psi_i[M,f,g]$. Also, let $M_j$ and $M_k$ be the restriction of $M'$ on $G_j$ and $G_k$, respectively. For every vertex $x\in X_i\cap V(M_j)\setminus V(M)$, set $f_{_1}(x)=1$ and $f_{_2}(x)=2$ and for every vertex $x\in X_i\cap V(M_k)\setminus V(M)$, set $f_{_1}(x)=2$ and $f_{_2}(x)=1$. Also, define $g_{_1}$ such that $g_{_1}(x)=g_{_1}(y)$ iff $x$ and $y$ are in the same connected component of $G_j[V(M_j)\cup f_{_1}^{-1}(2)]$ and define $g_{_2}$ similarly. One can check that $M_j$ and $M_k$ are solutions of $\psi_j[M,f_{_1},g_{_1}]$ and $\psi_k[M,f_{_2},g_{_2}]$, respectively. Also, it is clear that $(f_{_1},f_{_2})$ is a compatible pair with $f$ and $(g_{_1},g_{_2}) $ is an admissible pair with respect to $g$.
On the other hand, any two solutions for $G_j$ and $G_k$ can be combined to get a solution for $G_i$ which is a matching because $(f_{_1},f_{_2})$ is a compatible pair and is acyclic because $(g_{_1},g_{_2})$ is an admissible pair.  
Here, computation of $\psi_i[M,f,g]$ takes $t^{O(t)}$ time, since there are at most $2^{t}$ compatible pairs $(f_{_1},f_{_2})$ and at most $t^{O(t)}$ pairs $(g_{_1},g_{_2})$ and each pair can be checked to be admissible in poly($t$) time.

The runtime of the whole algorithm is at most $t^{O(t)}.O(n)$ because the number of possible matching $M$ is at most $t^{O(t)}$, the number of possible functions $f$ is at most $3^{t+1}$ and the number of possible functions $g$ is at most $(t+1)^{(t+1)}$. Also, the number of nodes in $T$ is at most $O(tn)$ \cite{cygan}.
\end{proof}

Using Theorem~\ref{theorem:8} and meta theorems about bidimensional problems \cite{demainehaji} (also see \cite{cygan}), one can deduce that  the  \textsc{Acyclic Matching} problem can be solved in time $k^{O(\sqrt{k})}\cdot n^{O(1)}$ for planar graphs and more generally for apex-minor-free graphs. However, this theory cannot be directly applied to find an FPT algorithm on $H$-minor-free graphs for an arbitrary fixed graph $H$, because the \textsc{Acyclic Matching} problem is not minor-closed in the sense that removing edges can increase the maximum acyclic matching (and also induced matching) of the graph (see \cite{cygan}). 
Alternatively, we use the induced grid theorem (Theorem~\ref{thm:nicol}) to prove that the induced and acyclic matching problems are both fixed-parameter tractable with respect to the parameter $k$ for any proper minor-closed class of graphs.
\begin{theorem}\label{thm:minor}
	Let $\mathcal{H}$ be a proper minor-closed class of graphs. The \textsc{Acyclic Matching} problem  and the \textsc{Induced Matching} problem are both fixed-parameter tractable with respect to the parameter $k$ on the class $\mathcal{H}$ .\label{theorem:14}
\end{theorem} 
To prove the theorem,  we need the following result from \cite{Nicola}. 

\begin{theorem} {\rm \cite{Nicola}} \label{thm:nicol}
	$($Induced grid theorem for minor-free graphs$)$. For every
	graph $H$ there is a function $f_H : \mathbb{N} \rightarrow \mathbb{N}$ such that every $H$-minor-free graph of tree-width at least $f_H(k)$ contains a $(k \times k)$-wall or the line graph of a chordless $(k \times k)$-wall as an induced subgraph.
\end{theorem}

\begin{proof}[Proof of Theorem~\ref{thm:minor}]
	Since $\mathcal{H}$ is a proper class of graphs, there is a graph $X$ with smallest number of vertices that is not in $\mathcal{H}$. Since $\mathcal{H}$ is minor-closed, every graph in $\mathcal{H}$ is  $X$-minor-free. By Theorem~\ref{thm:nicol}, there is a function $f_X:\mathbb{N}\rightarrow \mathbb{N}$ such that every $X$-minor-free graph of  tree-width at least $f_X(k)$ contains a $(k\times k)$-wall or the line graph of a chordless $(k\times k)$-wall as an induced subgraph. Now, let $(G,k)$ be an  instance of \textsc{Acyclic Matching} problem or \textsc{Induced Matching} problem, where $G\in \mathcal{H}$ and so is $X$-minor-free. If $\tw(G)\leq f_X(k)$,  the maximum acyclic matching of $G$ (because of Theorem~\ref{theorem:8}) and the  maximum induced matching of $G$ (because of Proposition 23 in \cite{PIM}) can be computed in FPT time with respect to the parameter $k$. So, assume that $\tw(G)\ge f_X(k)$. Theorem~\ref{thm:nicol} implies that $G$ contains a $(k\times k)$-wall or the line graph of a chordless $(k\times k)$-wall as an induced subgraph. Now, we prove that such $G$ has an induced (and so acyclic) matching of size $k$.  First, suppose that $G$ contains a $(k\times k)$-wall $W$ as an induced subgraph. The outer cycle $C$ of $W$ is an induced cycle in $G$ with length at least $6k-8$. 
	So, for $k\geq 3$,  $C$ contains an induced matching of size $k$ consisting of edges in $C$ with distance at least two from each other. Now, suppose that $G$ contains the line graph $W'$ of a chordless $(k \times k)$-wall as an induced subgraph. Again, the outer cycle $C'$ of $W'$ is an induced cycle of length at least $6k-8$, therefore $C'$ also contains an induced matching of size $k$, as desired.  
\end{proof}

In \cite{3}, Moser and Sikdar proved that the  \textsc{Induced Matching} problem admits a kernel  with $O(k\Delta^2)$ vertices (with a more attention at their proof, one may verify that, in fact, the problem admits a kernel with $O(k\Delta^2)$ edges). In other words, for every graph $G$ and integer $k$, either $G$ has an induced matching of size $k$, or $G$ has at most $O(k\Delta^2)$ edges. Since every induced matching is an acyclic matching, this immediately implies that  \textsc{Acyclic Matching} problem is  fixed-parameter tractable with respect to the parameters $k$ and $\Delta$ although it is $W[1]-$hard with respect to each of the two parameters $k$ and $\Delta$.

For the last result of the paper, we prove that the \textsc{Acyclic Matching} problem is fixed-parameter tractable with respect to the parameters $k$ and $c_4$, where $c_4$ is the number of cycles of length four in the input graph. The same result can be proved for the \textsc{Induced Matching} problem which improves a result by Moser et al. \cite{3} which states that the \textsc{Induced Matching} problem is fixed-parameter tractable with respect to the parameter $k$ for the graphs of girth at least $6$.

\begin{theorem}\label{thm:kc4}
	Let $c_{4}(G)$ be the number of cycles of length four in a graph $G$. The \textsc{Acyclic Matching} problem  and the \textsc{Induced Matching} problem are fixed-parameter tractable with respect to the parameters $k$ and $c_{4}$. In fact, it admits a polynomial kernel with $O(k(k^2+c_4)^4)$ edges.\label{theorem:12}
\end{theorem}
In order to prove the above theorem we need a result from Ramsey theory. Given graphs $H_1,H_2$, the Ramsey number $R(H_1,H_2)$ is the smallest integer $n$ such that for every graph $G$ on $n$ vertices, either $G$ contains a copy of $H_1$ as a subgraph, or the complement of $G$, $\overline{G}$, contains a copy of $H_2$. We need the following result from \cite{Ramsey1,Ramsey2} regarding the Ramsey number of the cycle $C_4$ versus  the complete graph $K_m$. There is a constant $c$, such that

\begin{equation}\label{ramsey}
R(C_4,K_m)\leq c (m/\log m)^2. 
\end{equation}

\begin{proof}[Proof of Theorem~\ref{thm:kc4}]
	Let $G$ be a graph with maximum degree $\Delta$. Without loss of generality, we may assume that every vertex of $G$ has at most one neighbor of degree one (otherwise, we can remove all leaf neighbors except one and this does not change the size of maximum induced and acyclic matching of $G$). Define $n_0=c(ck^2+4c_4)^2+4c_{4}$, where $c$ is the constant in \eqref{ramsey}. If $\Delta\leq n_0$, then by the argument before Theorem~\ref{thm:kc4}, we have a kernel with $O(k\Delta^2)=O(k(k^2+c_{4})^4)$ edges. Therefore, assume that $G$ has a vertex $v$ of degree at least $n_0+1$. In this case, we prove that $G$ is a yes-instance. Let $U=\{u_1,\cdots,u_{n_0}\}$ be the neighbors of $v$ with degree at least two. It is clear that there are at most $4c_4$ vertices in $U$ which are incident with a $4$-cycle $C_4$. So, there is a set $U'\subseteq U$ such that $ |U'|\geq c(ck^2+4c_4)^2 $ and the vertices of no $4$-cycle in $G$ intersects $U'$. Therefore, by \eqref{ramsey}, there is a stable set $I$ in $G[U']$ with $|I|\geq ck^2+4c_4$. Since every vertex of $I$ has degree at least two and no vertex in $U'$ is incident with a four-cycle in $G$, any $u_i\in I$ has a distinct private neighbor $w_i$. Set $W=\{w_i: u_i\in I\}$. With an argument similar to the one we gave for $U$, we can see that there is a stable set $I'\subseteq W$ with $|I'|\geq k$. Finally, define $M=\{u_j w_j:w_j\in I'\}$. We know that both $I$ and $I'$ are stable sets and no four-cycle in $G$ intersects vertices of $M$. Hence, $M$ is an induced matching of size at least $k$ which is an acyclic matching, as well. 
\end{proof}


\end{document}